\documentclass[sigconf, nonacm]{acmart}
\usepackage{url,caption}
\usepackage{tikz}
\usepackage{filecontents}
\usepackage{graphicx}
\usepackage{textcomp}
\usepackage{pifont,framed}
\usepackage[noend]{algpseudocode}
\usepackage{algorithm,algorithmicx}
\usepackage{xcolor}
\usepackage{color}
\usepackage{verbatim}
\usepackage[subfigure]{tocloft}
\usepackage{subfigure}
\usepackage[algo2e,ruled,vlined,lined]{algorithm2e}
\usepackage{multirow}
\usepackage{mathrsfs}
\usepackage{bm} 
\usepackage{algpseudocode}

\usepackage[font=small,skip=0pt]{caption}
\allowdisplaybreaks[4]
\def\BibTeX{{\rm B\kern-.05em{\sc i\kern-.025em b}\kern-.08em
		T\kern-.1667em\lower.7ex\hbox{E}\kern-.125emX}}

\newtheorem{definition}{\bf Definition}[section]
\newtheorem{theorem}{\bf Theorem}[section]

\usepackage{tikz}
\usepackage{eso-pic}

	\AddToShipoutPictureBG*{
		\begin{tikzpicture}[remember picture,overlay]
			\node[anchor=north,yshift=-0.3cm] at (current page.north) {%
				\textcolor{red}{\normalsize\textbf{Accepted by VLDB 2025}}%
			};
		\end{tikzpicture}%
	}

\begin{document}

\title{Privacy for Free: Leveraging Local Differential Privacy Perturbed Data from Multiple Services}

\author{Rong Du}
\email{roong.du@connect.polyu.hk}
\affiliation{%
	\institution{The Hong Kong Polytechnic University}
}

\author{Qingqing Ye}
\authornote{Corresponding author}
\email{qqing.ye@polyu.hk}
\affiliation{%
	\institution{The Hong Kong Polytechnic University}
}

\author{Yue Fu}
\email{yuesandy.fu@connect.polyu.hk}
\affiliation{%
	\institution{The Hong Kong Polytechnic University}
}

\author{Haibo Hu}
\email{haibo.hu@polyu.hk}
\affiliation{%
	\institution{The Hong Kong Polytechnic University}
}

\begin{abstract}Local Differential Privacy (LDP) has emerged as a widely adopted privacy-preserving technique in modern data analytics, enabling users to share statistical insights while maintaining robust privacy guarantees. However, current LDP applications assume a single service gathering perturbed information from users. In reality, multiple services may be interested in collecting users' data, which poses privacy burdens to users as more such services emerge. To address this issue, this paper proposes a framework for collecting and aggregating data based on perturbed information from multiple services, regardless of their estimated statistics (e.g., mean or distribution) and perturbation mechanisms.
	
Then for mean estimation, we introduce the Unbiased Averaging (UA) method and its optimized version, User-level Weighted Averaging (UWA). The former utilizes biased perturbed data, while the latter assigns weights to different perturbed results based on perturbation information, thereby achieving minimal variance. For distribution estimation, we propose the User-level Likelihood Estimation (ULE), which treats all perturbed results from a user as a whole for maximum likelihood estimation. Experimental results demonstrate that our framework and constituting methods significantly improve the accuracy of both mean and distribution estimation. 
\end{abstract}

\maketitle

\section{Introduction}
With the proliferation of the Internet, smart devices, and artificial intelligence, data has become one of our most valuable resources. Businesses across various sectors, especially healthcare, retailing, and finance, optimize their operations and decision-making by leveraging massive amounts of data. For instance, E-Commerce companies such as Amazon and Netflix rely on big data to provide personalized recommendations, enhancing customer experience and sales~\cite{hewage2018big,maddodi2019netflix}. However, privacy concerns grow as data scales exponentially and AI analysis techniques mature ~\cite{daswani2021equifax}.


To address these concerns, local differential privacy (LDP) \cite{chen2016private, duchi2013local, kasiviswanathan2011can}, a local variant of differential privacy model~\cite{dwork2008differential, dwork2006calibrating, mcsherry2007mechanism}, has widespread adoption in real-world applications. It allows privacy protection at the data source on the user's end to minimize the risk of privacy leakage, and thus helps enterprises to comply with increasingly stringent data protection regulations. As a result, leading tech companies embrace LDP in their mainstream products. For example, Google employs LDP to collect user data on Chrome home pages \cite{erlingsson2014rappor}, while Apple utilizes it on iOS devices to analyze app engagement patterns privately \cite{team2017learning}. Similarly, Samsung incorporates LDP into its Android phones to securely gather telemetry data \cite{nguyen2016collecting}, and Huawei implements differential privacy techniques in various products and services to protect users' personal information \cite{keyofhuawei}.

\begin{figure}
	\centering
	\includegraphics[width=0.45\textwidth]{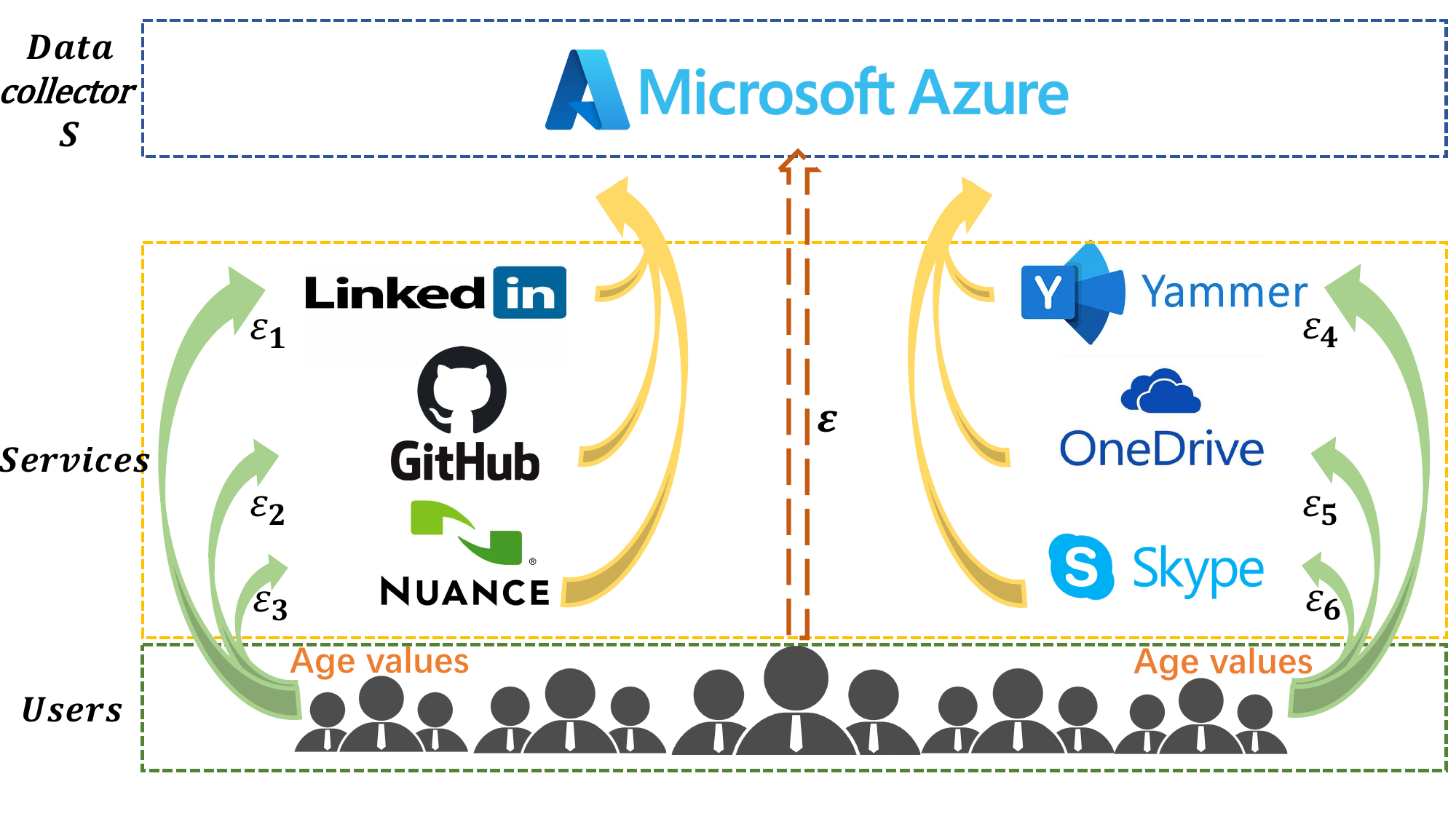}
	\caption{An example for multiple services under LDP}
	\label{fig:example1}
	\vspace{-0.2in}
\end{figure}


Many existing LDP applications assume a standalone service where each user's information is collected only once. However, in the era of big data, data are often collected and processed by multiple services early on. This is common in government statistical systems, which aggregate data from multiple agencies or organizations in the field of health, education, justice, and social welfare, such as the U.S. Federal Data Strategy \footnote{https://strategy.data.gov/action-plan/} and the European Statistical System \footnote{https://ec.europa.eu/eurostat/web/main/publications/statistical-reports}. In the healthcare domain, medical institutions independently collect patient data, as demonstrated by PCORnet ~\cite{toh2017national}, which involves multiple healthcare organizations collecting similar patient data, and the FDA's Sentinel Initiative~\cite{platt2012us}, which gathers data from multiple healthcare providers to monitor drug safety. Public health authorities like CDC can leverage existing data from these institutions for population-level analyses by using perturbed data\footnote{https://www.cdc.gov/mmwr/pdf/other/su6103.pdf}, reducing privacy risks and patient burden.

The multi-service data collection scenario is illustrated in Figure \ref{fig:example1}, where various services in the yellow box independently collect users' age information. In this scenario, users need to submit their age information multiple times, spending a privacy budget for each submission. Suppose a new service (e.g., Microsoft Azure), referred to as data collector $S$, aims to acquire users' age information. If $S$ directly requests users to upload their data, according to the composition theorem \cite{li2016differential}, the users' privacy would be further compromised by increasing the overall privacy budget by an additional $\epsilon$. As an alternative, $S$ can leverage the existing perturbed data from other services and thus keep the privacy budget of all users unaltered. Moreover, since the perturbed data gathered by multiple services is abundant, data collector $S$ can aggregate more accurate statistical results than from any single service.

In this paper, we investigate such a scenario where multiple services gather identical numerical data from users with the same or different mechanisms and privacy budgets. Unlike prior studies such as \cite{yiwen2018utility, du2023differential}, which only consider mean aggregation from a single mechanism using different privacy budgets, our study explores both mean and generic distribution estimation under various mechanisms and privacy budgets. For collecting mean value, various perturbation mechanisms exist (e.g., Laplace ~\cite{dwork2006calibrating}, Stochastic Rounding~\cite{duchi2018minimax}, Piecewise Mechanism~\cite{wang2019collecting}). As for distribution estimation, state-of-the-art perturbation mechanisms include Piecewise Mechanism~\cite{wang2019collecting} and Square Wave~\cite{li2020estimating}). 
By leveraging all perturbed results provided by different services, we can prevent further privacy leakage of user data while also enhancing data utility. However, aggregating statistical results from different services poses several challenges:


\begin{itemize}
	\item
	Perturbation mechanisms for numerical data are designed for different statistics --- some for mean, and others for distribution estimation. This raises the question of how to use distribution estimation mechanisms for mean estimation and vice versa.
	\item
	Perturbed results from different mechanisms fall in various ranges, and some are even biased, making them unsuitable for direct use in mean estimation. It is an open question how to maximize their utility for the best aggregate results.
\end{itemize}

To address these challenges, we first propose the Unbiased Averaging (UA) method for mean estimation, which converts biased perturbed results into unbiased ones before averaging. Based on this, we introduce an optimized method called User-level Weighted Averaging (UWA), which employs Bayesian updating~\cite{hacking1967slightly} to infer the distribution of original values and assigns different weights to their perturbed values, which obtains an optimized mean with minimal variance. Additionally, we propose a general method for distribution estimation, namely User-level Likelihood Estimation (ULE), which treats all perturbed results from a user as a whole for maximum likelihood estimation to infer the original distribution. By exploiting the correlation between perturbed values from the same user, our aggregation method extracts information about the original values and is agnostic to perturbation mechanisms, privacy budget settings, and the number of services. The main contributions of this paper are summarized as follows:
\begin{itemize}
	\item
	This is the first work to investigate a new data collection and aggregation paradigm in the context of LDP, where a data collector gathers and aggregates perturbed values from multiple services. 
	
	\item
	Under this paradigm, we propose Unbiased Averaging (UA) and User-level Weighted Averaging (UWA) for mean estimation, and User-level Likelihood Estimation (ULE) for distribution estimation. These methods can integrate all perturbed information to aggregate accurate mean or distribution statistics.
	
	\item
	We conduct extensive experiments, and show that our approach significantly improves the performance of any single mechanism in either mean or distribution estimation.
\end{itemize}

The rest of this paper is organized as follows. In Section \ref{preliminaries}, we introduce some preliminaries of LDP and existing LDP perturbation mechanisms for numerical data. Section \ref{problemdefinition} formally presents the system framework and problem definition. In Section \ref{Tech1}, we propose the mean estimation method, followed by the distribution estimation method in Section \ref{Tech2}. Section \ref{exp} provides extensive experimental evaluations. We then review related studies in Section \ref{relatedwork} and conclude the paper in Section \ref{conclusion}. 
\section{Preliminaries}
\label{preliminaries}

\subsection{Local Differential Privacy}
LDP is a promising privacy protection technique that enables users to perturb their data before publishing. It is defined as follows:

\begin{definition}
	A randomized algorithm $\mathcal{A}$ satisfies $\epsilon$-local differential privacy ($\epsilon$-LDP), if and only if for any two values $x$ and $x'$, and any possible output $y \in \text{Range}(\mathcal{A})$, the following condition holds:
	\begin{equation*}
		\setlength{\abovedisplayskip}{1pt} 
		\setlength{\belowdisplayskip}{1pt} 
		\frac{P[\mathcal{A}(x)=y]}{P[\mathcal{A}(x')=y]} \leq e^{\epsilon}, \epsilon > 0.
	\end{equation*} 
\end{definition}

The core idea of LDP lies in its probabilistic nature: mechanism $\mathcal{A}$ maps any specific input to an output based on a probability distribution. This distribution is controlled by the privacy budget $\epsilon$. As a result, no one can determine an individual's true input by observing an output $y$. The privacy budget $\epsilon$, controls the trade-off between privacy protection and data utility. A smaller $\epsilon$ provides stronger privacy guarantees but potentially reduces the accuracy of data analysis, while a larger $\epsilon$ allows for more accurate analysis at the cost of reduced privacy protection.

\subsection{Laplace Mechanism}
Laplace mechanism~\cite{dwork2006calibrating} is a commonly used differential privacy mechanism that achieves privacy protection by adding Laplace noise to the original data. The input value $v$ is always normalized into [-1, 1]. The Probability Density Function(PDF) of the perturbed result $\mathcal{A}_{lap}(v)$ can be expressed as:
\begin{equation*}
	\setlength{\abovedisplayskip}{1pt} 
	\setlength{\belowdisplayskip}{1pt} 
	\rho(\mathcal{A}_{lap}(v) = x | v) = \frac{1}{2b} \exp(-\frac{|x - v|}{b}),
\end{equation*}
where $b = \frac{2}{\epsilon}$. This distribution has a mean of $v$, so the estimate $\mathcal{A}_{lap}(v)$ is unbiased. In addition, the variance is $var_{lap}(v)=\frac{8}{\epsilon^2}$. If the service receives $n$ perturbed values, it simply computes their average $\frac{1}{n} \sum_{i=1}^n \mathcal{A}_{lap}(v_i)$ as an estimate of the original mean.

\subsection{Stochastic Rounding (SR)}

Stochastic Rounding (SR)~\cite{duchi2018minimax}, also known as the Duchi et al. solution, the core principle is rooted in the Bernoulli distribution, resulting in outputs that are both bounded and discrete.The input domain is typically normalized to [-1, 1], and the output value follows a two-point distribution. The probability mass function (PMF) of the output $\mathcal{A}_{sr}(v)$, given input $v$, can be formally written as:
\begin{equation}
	\setlength{\abovedisplayskip}{1pt} 
	\setlength{\belowdisplayskip}{1pt} 
	P(\mathcal{A}_{sr}(v) = x | v) = \left\{
	\begin{array}{ll}
		p, & \text{if } x = \frac{e^\epsilon + 1}{e^\epsilon - 1} \\
		1-p, & \text{if } x = -\frac{e^\epsilon + 1}{e^\epsilon - 1}\\
	\end{array}
	\right.
	\label{SR}
\end{equation}

where the probability $p$ is given by:
\begin{equation*}
	\setlength{\abovedisplayskip}{1pt} 
	\setlength{\belowdisplayskip}{1pt} 
	p = \frac{e^\epsilon - 1}{2e^\epsilon + 2} \cdot v + \frac{1}{2}.
\end{equation*}

In SR mechanism, $\mathcal{A}_{sr}(v)$ is an unbiased estimator of the input value $v$. In addition, the variance is  $var_{sr}(v)=(\frac{e^\epsilon+1}{e^\epsilon-1})^2 - v^2$.

\subsection{Piecewise Mechanism (PM)}
Different from SR, PM~\cite{wang2019collecting} perturbs values into a domain instead of two values. Given input value $v$ $\in$ $[-1,1]$, the PDF of output $\mathcal{A}_{pm}(v)$ $\in$ $[-C,C]$ has two parts: domain $[l(v),r(v)]$ and domain $[-C,l(v))\cup(r(v),C]$, where $C=\frac{e^{\epsilon/2}+1}{e^{\epsilon/2}-1}$, $l(v)=\frac{C+1}{2}v-\frac{C-1}{2}$ and $r(v)=l(v)+C-1$. Given input $v$, the perturbed value is in the range $[l(v),r(v)]$ with high probability and in the range $[-C,l(v))\cup(r(v), C]$ with low probability. 

The conditional probability density function of the PM, denoted as $\rho(\mathcal{A}_{pm}(v) = x | v)$, is given by:
\begin{equation*}
	\setlength{\abovedisplayskip}{1pt} 
	\setlength{\belowdisplayskip}{1pt} 
	\rho(\mathcal{A}_{pm}(v) = x | v) = \left\{
	\begin{array}{ll}
		\frac{e^{\epsilon}-e^{\epsilon/2}}{2(e^{\epsilon/2}+1)}, & \text{if } x \in [l(v), r(v)] \\[10pt]
		\frac{e^{\epsilon/2}-1}{{2(e^{\epsilon/2}+e^{\epsilon})}}, & \text{if } x \in [-C,l(v)) \cup (r(v),C] 
	\end{array}
	\right.
\end{equation*}

Perturbed value $\mathcal{A}_{pm}(v)$ is an unbiased estimator of input value $v$, the service can use the mean of collected values as an unbiased estimator of the mean of input values.  Moreover, the variance is
\begin{equation*}
	\setlength{\abovedisplayskip}{1pt} 
	\setlength{\belowdisplayskip}{1pt} 
	var_{pm}(v)= \frac{v^2}{e^{\epsilon/2}-1} + \frac{e^{\epsilon/2}+3}{3(e^{\epsilon/2}-1)^2}.
\end{equation*}

\subsection{Square Wave Mechanism (SW)}
\label{SW}
SW~\cite{li2020estimating} is designed for distribution estimation, and the main idea is to increase the probability of output value that can provide more information about the input value. The service receives perturbed values from users and reconstructs the distribution over a discrete numerical domain. Each user processes a floating value in the domain [0,1] and generates an output value in $[-b, 1+b]$, where $b=\frac{\epsilon e^{\epsilon}-e^{\epsilon}+1}{2e^{\epsilon}(e^{\epsilon}-\epsilon-1)}$. The conditional probability density function of the SW, denoted as $\rho(\mathcal{A}_{sw}(v) = x | v)$, is given by:
\begin{equation*}
	\setlength{\abovedisplayskip}{1pt} 
	\setlength{\belowdisplayskip}{1pt} 
	\rho(\mathcal{A}_{sw}(v) = x | v) = \left\{
	\begin{array}{ll}
		\frac{e^\epsilon}{2be^\epsilon+1}, & \text{if } x \in [v-b, v+b] \\[10pt]
		\frac{1}{2be^\epsilon+1}, & \text{if } x \in [-b,v-b) \cup (v+b,1+b],
	\end{array}
	\right.
\end{equation*}
where $p=\frac{e^\epsilon}{2be^\epsilon+1}$ and $q=	\frac{1}{2be^\epsilon+1}$. After reviving the perturbed values, the service aggregates the original distribution by using the Maximum Likelihood Estimation (MLE)~\cite{rossi2018mathematical} and reconstructs the distribution of original values. $\mathcal{A}_{sw}(v)$ is not an unbiased estimator of $v$, and its variance is:
\begin{equation*}
	\setlength{\abovedisplayskip}{1pt} 
	\setlength{\belowdisplayskip}{1pt} 
	\begin{split}
		var_{sw}(v) &=\frac{q(b^3-(b+v)^3-(b-v)^3+(b+1)^3)}{3} \\& - \frac{(q + 2bq + 4bpv - 4bqv)^2}{4} + \frac{2bp(b^2 + 3v^2)}{3}
	\end{split}
\end{equation*}

\section{System Framework and Problem Definition}
\label{problemdefinition}

Consider a system with $m$ services which have already collected user data. Each service $j$ employs a perturbation mechanism $\mathcal{A}_j$, which consumes a privacy budget of $\epsilon_j$. It's important to note that these perturbation mechanisms can be the same across different services or completely different. There are $n$ users and the $i$-th user has an original value $v_i$. For user $i$, the value $v_i$ is perturbed by mechanism $\mathcal{A}_j$, resulting in the perturbed value $\mathcal{A}_j(v_i)$ being sent to the $j$-th service. From a user's perspective, the set of data transmitted to all services is represented as $U_i = \{\mathcal{A}_1(v_i), \dots, \mathcal{A}_m(v_i)\}$. Conversely, from a service's viewpoint, the $j$-th service receives a set of perturbed data from all users, denoted as $C_j =\{\mathcal{A}_j(v_1), \dots, \mathcal{A}_j(v_n)\}$.

In addition to the $m$ services that have already collected data sets, we also consider a data collector $S$. $S$ will collect data from these services and aggregate the perturbed data sets to produce an accurate statistical result\footnote{This data collector $S$ can be any one of the $m$ services, and it can collaborate with the other $m-1$ services to enhance performance.}. The specific process is illustrated in Figure \ref{fig:system}. All users perturb their data according to the given perturbation mechanism $\mathcal{A}_j$ with $\epsilon_j$ and then upload the perturbed value to the $j$-th service (Step \textcircled{1}). Each service can perform its own statistical analysis to obtain the desired results. $S$ collects the perturbed data set $\{C_1,\dots, C_m\}$ from services (Step \textcircled{2}). Finally, $S$ utilizes these data to aggregate an accurate mean and distribution estimation (Step \textcircled{3}).

\begin{figure}
	\centering
	\includegraphics[width=0.5\textwidth]{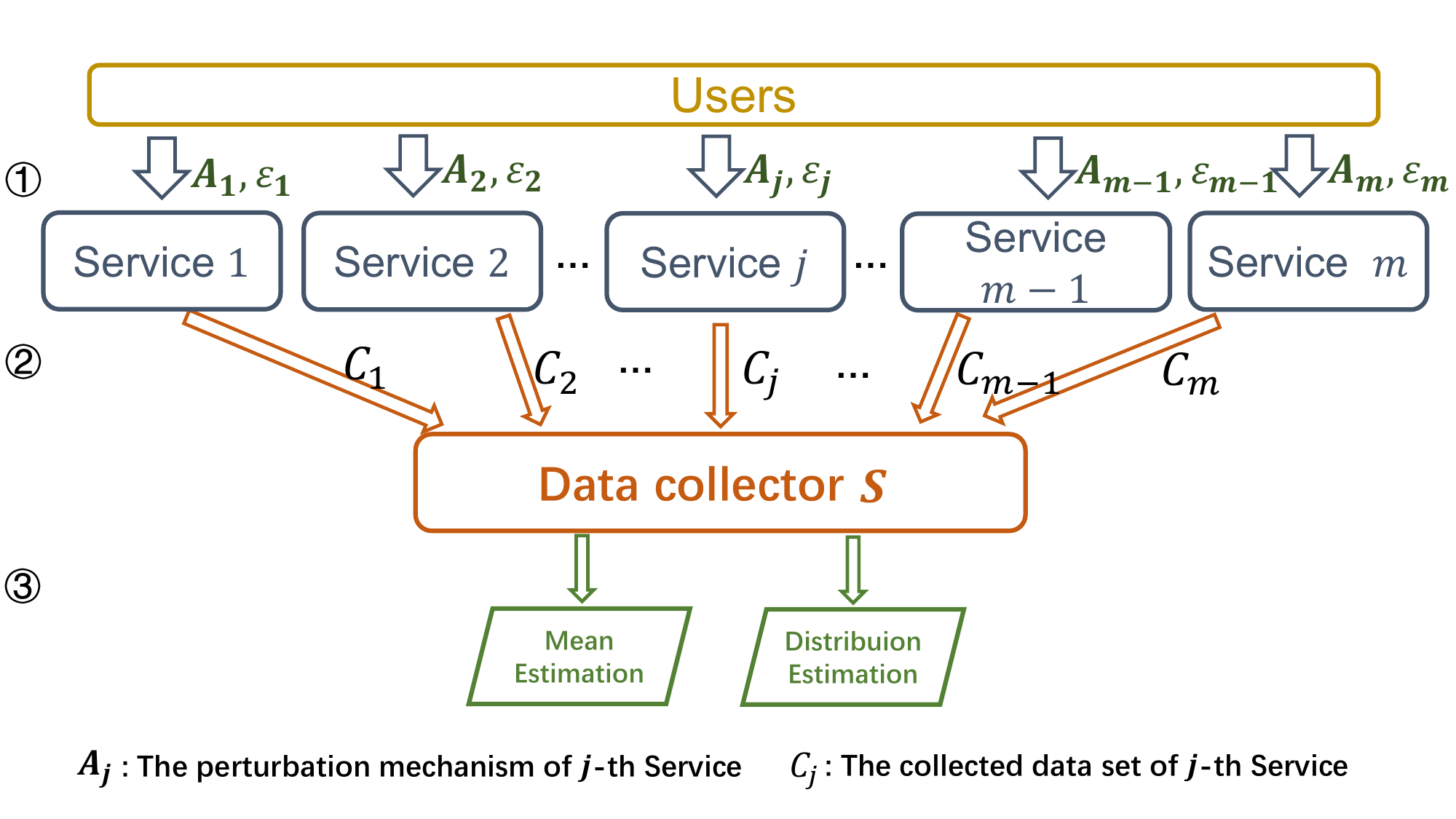}
	\caption{The multiple service data collecting system}
	\label{fig:system}
\end{figure}

\noindent\textbf{Mean estimation for multiple services.}
Let $M$ denote the mean of the original values from all users, which can be expressed as:
\begin{equation*}
	\setlength{\abovedisplayskip}{1pt} 
	\setlength{\belowdisplayskip}{1pt} 
	M = \frac{1}{n} \sum_{i=1}^{n} v_i.
\end{equation*}
Data collector $S$ accesses the maximum amount of perturbed information, which can be represented as the set $\{C_1, \dots, C_m\}$. This set, representing data sent by all users, can also be expressed as $\{U_1,\dots,U_n\}$.

The first task is to find an algorithm $\mathbb{A}_m$ that, based on the perturbed information available, estimates the mean $\hat{M}$ as follows:
\begin{equation*}
	\setlength{\abovedisplayskip}{1pt} 
	\setlength{\belowdisplayskip}{1pt} 
	\hat{M} =  \mathbb{A}_m(C_1, \dots, C_m) =  \mathbb{A}_m(U_1, \dots, U_n).
\end{equation*}

Our objective is to ensure that the estimated mean $\hat{M}$ is as close as possible to $M$. To evaluate the quality of our mean estimation, we primarily use the Mean Squared Error (MSE)~\cite{Mood1963Introduction}.

\noindent\textbf{Distribution estimation for multiple services.}
In this paper, we will also investigate the distribution estimation of numerical data. For practical purposes, such as computational efficiency and ease of analysis, we often approximate a continuous distribution with a discrete one. This involves discretizing the original values' distribution into $l$ buckets, resulting in the discrete ground truth distribution $D = \{d_1, \dots, d_l\}$.

Similar to the setting for mean estimation, data collector $S$ possesses the maximum amount of perturbed information represented as $\{C_1, \dots, C_m\}$. We need to propose a distribution estimation algorithm $\mathbb{A}_d$ to estimate the distribution $D$. This estimation can be expressed as:
\begin{equation*}
	\setlength{\abovedisplayskip}{1pt} 
	\setlength{\belowdisplayskip}{1pt} 
	\hat{D} = \mathbb{A}_d(C_1, \dots, C_m) = \mathbb{A}_d(U_1, \dots, U_n).
\end{equation*}
where $\hat{D}$ is the estimated discrete density distribution. The probability associated with each bucket in $\hat{D}$ is given by:
\begin{equation*}
	\setlength{\abovedisplayskip}{1pt} 
	\setlength{\belowdisplayskip}{1pt} 
	\hat{D} = \{\hat{d}_1, \dots, \hat{d}_t, \dots, \hat{d}_l\},
\end{equation*}
and these probabilities must satisfy the condition:
\begin{equation*}
	\setlength{\abovedisplayskip}{1pt} 
	\setlength{\belowdisplayskip}{1pt} 
	\sum_{t=1}^{l} \hat{d}_t = 1, \quad \hat{d}_t > 0 \quad \text{for all } t = 1, 2, \ldots, l
\end{equation*}

Our goal is to achieve a distribution estimate $\hat{D}$ that closely approximates the true distribution $D$. The quality of this approximation is evaluated using either the Kullback-Leibler (KL) divergence \cite{manning1999foundations} or the Jensen-Shannon (JS) divergence \cite{dagan1997similarity}.

\section{Mean Estimation}
\label{Tech1}
A naive way to estimate the mean $\hat{M}$ is to average all the perturbed results directly. However, this approach typically requires the perturbed values to be unbiased. Therefore, in Section 4.1, we propose a method to convert biased mechanisms into unbiased ones to facilitate averaging. Subsequently, Section 4.2 and Section 4.3 introduce an optimization method that assigns weights to each perturbed result based on its variance, thereby achieving better utility.

\subsection{Unbiased Averaging (UA)}
When all the perturbed results are unbiased, i.e., the expected output is equal to the input value, the direct method for mean estimation is to average all the collected data directly, as follows:
\begin{equation}
	\setlength{\abovedisplayskip}{0pt} 
	\setlength{\belowdisplayskip}{0pt} 
	\label{UA}
	\hat{M} = \frac{\sum_{j=1}^{m} \sum_{i=1}^{n} \mathcal{A}_j(v_i)}{n \times m}.
\end{equation}

Not all perturbation mechanisms are unbiased - the SW mechanism is a notable example of a biased mechanism for mean estimation. To effectively utilize these biased perturbed values, it is necessary to transform them into unbiased ones in order to accurately estimate the mean. We propose an unbiased method, which is presented in Theorem \ref{unbiase_theory}.

\begin{theorem}
	\label{unbiase_theory}
	\textbf{(Unbiasing Process)} Given a perturbation mechanism $\mathcal{A}$ and an input $v$, the perturbed value follows a distribution related to $v$, and the expectation of this perturbed value is denoted as $E_\mathcal{A}(v) = y$. If $y \neq v$, we say the mechanism is biased. Assume that the function $E_\mathcal{A}(v)$ is invertible, i.e., there exists a function $E_\mathcal{A}^{-1}$ such that $E_\mathcal{A}^{-1}(y) = v$. Then, $E_\mathcal{A}^{-1}(y)$ is an unbiased estimation of the original input $v$.
\end{theorem}

\begin{proof}
	To prove that $E_\mathcal{A}^{-1}(y)$ is an unbiased estimation of the original input $v$, we need to show that $E[E_\mathcal{A}^{-1}(y)] = v$. The proof proceeds as follows:
	\begin{align*}
		\setlength{\abovedisplayskip}{0pt} 
		\setlength{\belowdisplayskip}{0pt} 
		E[E_\mathcal{A}^{-1}(y)] &= E[E_\mathcal{A}^{-1}(E_\mathcal{A}(v))] = E[v]
	\end{align*}
	
	Since the expectation of a random variable is equal to its expected value, i.e., $E[v] = v$. Therefore, $E_\mathcal{A}^{-1}(y)$ is an unbiased estimator of $v$.
\end{proof}

\noindent\textbf{Case study.} We present an example of an unbiasing process for the SW mechanism. For the original data $v$, the expectation of $\mathcal{A}_{sw}(v)$ can be calculated as follows:
\begin{align*}
	E(\mathcal{A}_{sw}(v)) 
	& = \int_{-b}^{v-b} xq \, dx + \int_{v-b}^{v+b} xp \, dx + \int_{v+b}^{1+b} xq \, dx \\
	&= q \left[ \frac{x^2}{2} \right]_{-b}^{v-b} + p \left[ \frac{x^2}{2} \right]_{v-b}^{v+b} + q \left[ \frac{x^2}{2} \right]_{v+b}^{1+b} \\
	&= \frac{q}{2}+qb+ 2b(p-q)v.
\end{align*}
where $b,p,q$ are listed in Section \ref{SW}. Let $\mathcal{A}{usw}(v)$ denote the unbiased result of $v$ using the unbiased SW method. According to Theorem \ref{unbiase_theory}, $\mathcal{A}{usw}(v)$ can be expressed as:
\begin{equation*}
	\setlength{\abovedisplayskip}{0pt} 
	\setlength{\belowdisplayskip}{0pt} 
	\mathcal{A}_{usw}(v)= \frac{\mathcal{A}_{sw}(v)-q/2-qb}{2b(p-q)}.
\end{equation*}
The variance for unbiased SW is:
\begin{equation*}
	\setlength{\abovedisplayskip}{0pt} 
	\setlength{\belowdisplayskip}{0pt} 
	var_{usw}(v)= \frac{var_{sw}(v)}{(2b(p-q))^2}.
\end{equation*}

\subsection{Posterior Distribution Estimation for Individual Users}
In the UA algorithm, directly averaging all perturbed results assumes an equal contribution from each perturbed value. However, to achieve better utility, mechanisms with lower variance should be assigned higher weights. The challenge lies in determining the variance of each mechanism's estimation results without prior knowledge of the original distribution. Moreover, the variance heavily relies on the privacy budget $\epsilon$, original value $v$, and the perturbation mechanisms themselves.

Figure \ref{variance} illustrates the variance for several common LDP perturbation mechanisms. Specifically, Figure \ref{variance}(a) depicts the variance with varying $v$ when $\epsilon = 1$, while Figure \ref{variance}(b) shows the variance with varying $\epsilon$ when $v = 1$. Although these two figures do not encompass all possible variance scenarios, they demonstrate that the variance of a single mechanism changes with both $v$ and $\epsilon$. Furthermore, the relative ranking of variances among different mechanisms shifts under different parameter settings.

\begin{figure}[th]
	\centering
	\vspace{-0.05\linewidth}
	\hspace{-0.05\linewidth}
	\subfigure[$\epsilon=1$. \label{fig:subfig1}]{
		\centering
		\includegraphics[width=0.49\linewidth]{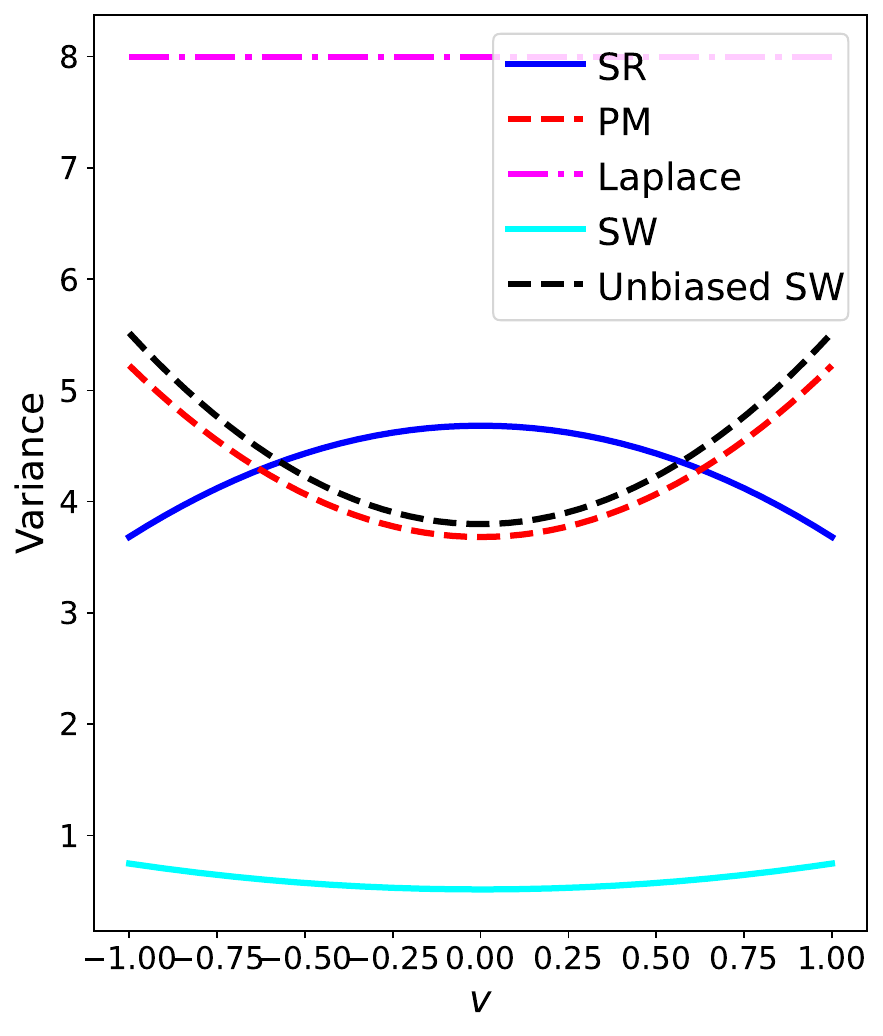}
	}
	\subfigure[$v=1$. \label{fig:subfig2}]{
		\centering
		\includegraphics[width=0.49\linewidth]{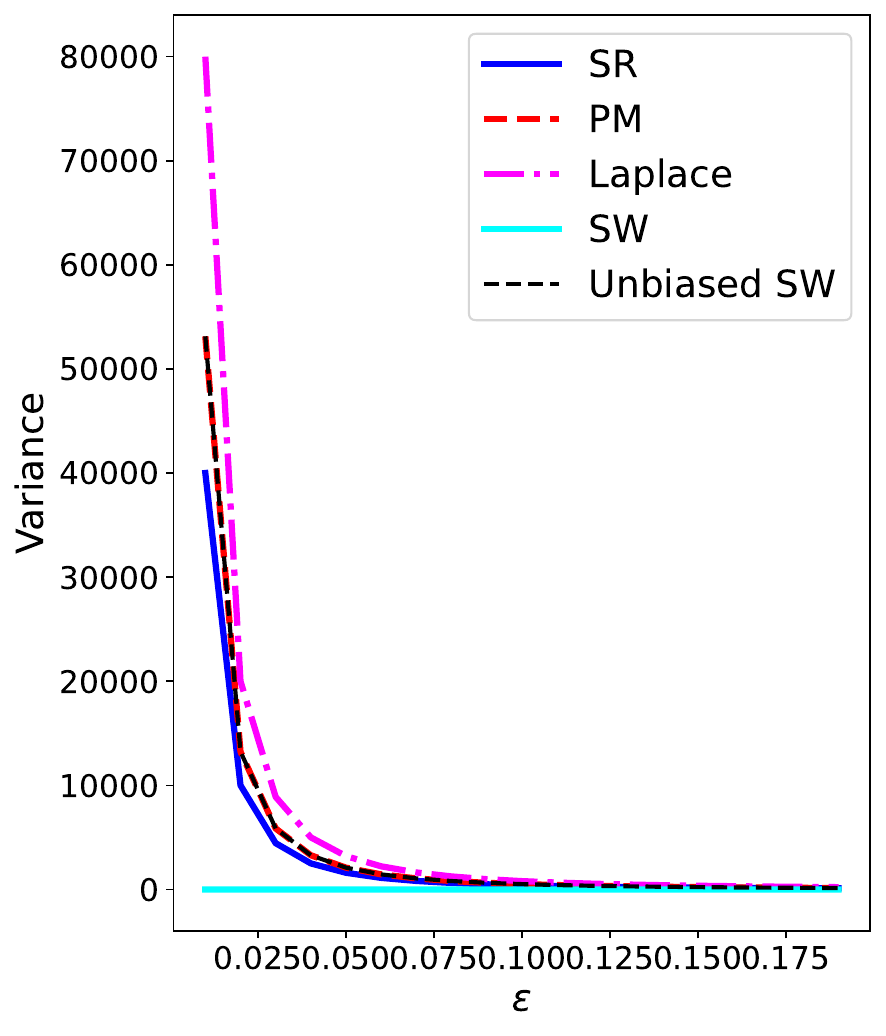}
	}
	\caption{Variance comparison with varying $v$ and $\epsilon$}
	\label{variance}
	\vspace{-0.12in}
\end{figure}

If we knew the ground truth of $v$ and the corresponding $\epsilon$, we could select the optimal mechanism to perturb the data and thus achieve minimum variance. However, this approach is impractical. Firstly, $v$ is unknown in real-world scenarios. Additionally, selecting only one mechanism's result would waste information from other perturbation results. Fortunately, since multiple perturbed values are transmitted from the same user, we can perform Bayesian inference to infer the PDF of the original value $v$ based on these perturbed values as its posterior distribution \cite{box2011bayesian}. This inference enables us to assign optimal weights to different perturbed results, ultimately leading to a minimal variance for more accurate mean estimation. Next, we will describe how we estimate the posterior distribution of each user's value in this section. Then, in Section 4.3, we will explain how to allocate weights based on the estimated posterior distributions for individual users.

\noindent\textbf{Bayesian Updating for Posterior Distribution Estimation.}
Let us consider a single user's original data, denoted as $v$, which corresponds to $m$ observed perturbed values $\{\mathcal{A}_1(v), \ldots, \mathcal{A}_m(v)\}$. Our goal is to infer the posterior distribution of $v$ using these observed values through Bayesian updating. The utilization of each observed value shall be referred to as an update.

Let $f_j(v)$ represent the current estimate of the posterior distribution of $v$ after updating $j$ times. Given the absence of prior information about $v$, we initialize its prior distribution $f_0(v)$ as a uniform distribution. We start by considering how to update the distribution using only one observed value, $\mathcal{A}_1(v)$. 

Let $P(\mathcal{A}_1(v) \mid v)$ denote the probability of obtaining the perturbed value $\mathcal{A}_1(v)$ given the original value $v$. Using Bayes' theorem, we can update our prior belief $f_0(v)$ based on the observed $\mathcal{A}_1(v)$, obtaining the posterior probability $f_1(v)$ as follows:
\begin{equation}
	\setlength{\abovedisplayskip}{0pt} 
	\setlength{\belowdisplayskip}{0pt} 
	f_1(v)=P(v \mid \mathcal{A}_1(v)) = \frac{P(\mathcal{A}_1(v) \mid v) f_0(v)}{P(\mathcal{A}_1(v))},
	\label{PP1}
\end{equation}
where
\begin{equation}
	\setlength{\abovedisplayskip}{0pt} 
	\setlength{\belowdisplayskip}{0pt} 
	P(\mathcal{A}_1(v)) = \int P(\mathcal{A}_1(v) \mid v) f_0(v) dv.
	\label{PM1v}
\end{equation}

To facilitate computation, we discretize Equation \ref{PM1v}. We uniformly divide the input domain into $h$ buckets $\{B_1,\dots,B_k, \dots, B_h\}$, with $\mu_k$ denoting the midvalue of the $k$-th bucket. The posterior distribution can be represented as a set of pairs, each containing a bucket midvalue and its corresponding probability estimate: $f_j(v) = {(\mu_1, f_j^1), \dots, (\mu_k, f_j^k), \dots, (\mu_h, f_j^h)}$ ($j\in\{0,\dots,m\}$). Consequently, Equation \ref{PM1v} can be rewritten as:
\begin{equation*}
	\setlength{\abovedisplayskip}{0pt} 
	\setlength{\belowdisplayskip}{0pt} 
	P(\mathcal{A}_1(v)) = \sum_{k=1}^h P(\mathcal{A}_1(v) \mid \mu_k) f_0^k.
\end{equation*}

To calculate $P(\mathcal{A}_1(v)\mid \mu_k)$, we discretize the output domain by uniformly partitioning it into multiple buckets in a similar manner. For clarity, we denote $\mathcal{A}_1(v)\in B'$ as the event that $\mathcal{A}_1(v)$ falls into a specific bucket $B'$. Let $s$ be the width of bucket $B'$, with $\mu'$ representing its midpoint. We can then calculate $P(\mathcal{A}_1(v)\in B' \mid \mu_k)$ as follows:
\begin{equation}
	\setlength{\abovedisplayskip}{0pt} 
	\setlength{\belowdisplayskip}{0pt} 
	P(\mathcal{A}_1(v)\in B' \mid \mu_k )= s\cdot \rho_{\mu_k,\mathcal{A}_1}(\mu'),
	\label{Bucket}
\end{equation}
where $\rho_{\mathcal{A}_1,\mu_k}(\mu')$ denotes the probability density at point $\mu'$. This probability density function is determined by the input value $\mu_k$ and the mechanism $\mathcal{A}_1$. The probability density functions for different mechanisms are detailed in Section 2. To calculate $P(\mathcal{A}_1(v) \mid v)$, we need to find the nearest bucket midvalue $\mu_t$ to the input $v$. The calculation of $P(\mathcal{A}_1(v) \mid v)$ is then equivalent to $P(\mathcal{A}_1(v) \mid \mu_t)$.

As the number of observations increases, we can gradually update the posterior distribution. The result after updating with the first $j-1$ observations is:
\begin{equation*}
	\setlength{\abovedisplayskip}{0pt} 
	\setlength{\belowdisplayskip}{0pt} 
	\begin{split}
		f_{j-1}(v) =P(v|\mathcal{A}_1(v), ..., \mathcal{A}_{j-1}(v)).
	\end{split}
\end{equation*}

Based on the $(j-1)$-th observed value $\mathcal{A}_{j-1}$, we update and obtain $f_j(v)$ as follows:
\begin{equation}
	\setlength{\abovedisplayskip}{0pt} 
	\setlength{\belowdisplayskip}{0pt} 
	f_j(v)=P(v|\mathcal{A}_1(v), ..., \mathcal{A}_j(v)) = \frac{P(\mathcal{A}_j(v)|v)f_{j-1}(v)}{P(\mathcal{A}_j(v)|\mathcal{A}_1(v), ..., \mathcal{A}_{j-1}(v))},
	\label{PM2v}
\end{equation}
where
\begin{equation*}
	\setlength{\abovedisplayskip}{0pt} 
	\setlength{\belowdisplayskip}{0pt} 
	\begin{split}
		P(\mathcal{A}_j(v)|\mathcal{A}_1(v), ..., \mathcal{A}_{j-1}(v)) = \sum_{k=1}^h P(\mathcal{A}_j(v)|\mu_k))f_{j-1}(\mu_k).
	\end{split}
\end{equation*}

Algorithm \ref{BayesianUpdating} outlines the procedure for estimating the posterior distribution of a single value $v$. The process begins by initializing the prior distribution of $v$ as a uniform distribution (line 1). Subsequently, the algorithm sequentially processes $m$ perturbed observations. Each iteration calculates the conditional probability of the current observation and updates the posterior distribution of $v$ using Bayes' theorem (line 2). This iterative process continues until all $m$ observations have been processed. The algorithm concludes by returning the final updated posterior distribution of $v$, denoted as $f_m(v)$ (line 3).

\begin{algorithm}
	\caption{Bayesian updating}
	\begin{flushleft}
		{\bf Input:}
		Original data $v$; the observed perturbed results $\{\mathcal{A}_1(v), \mathcal{A}_2(v), \ldots, \mathcal{A}_m(v)\}$, the size of discrete buckets for distribution estimation $h$.
		\\
		{\bf Output:}
		Updated posterior distribution of $v$: $f_m(v)$
	\end{flushleft}
	\begin{algorithmic}[1]
		\State Initialize $f_0(v) = \{f_0^1(v), \ldots, f_0^k(v),\dots,f_0^h(v)\}, f_0^k(v)=\frac{1}{h}$
		\State\For{$j = 1$ to $m$}
		{\begin{equation*}
				\begin{split}
					&p(\mathcal{A}_j(v)|\mathcal{A}_1(v),  \! \ldots \! , \mathcal{A}_{j-1}(v))= \sum_{j=1}^h P(\mathcal{A}_j(v)|v) \cdot f_{j-1}(v)
					\\
					&f_j(v) = P(v|\mathcal{A}_1(v),  \! \ldots \! , \mathcal{A}_i(v)) \!  =  \! \frac{P(\mathcal{A}_j(v)|v) \cdot f_{j-1}(v)}{P(\mathcal{A}_j(v)|\mathcal{A}_1(v), \ldots, \mathcal{A}_{j-1}(v))}
				\end{split}
			\end{equation*}
		}
		\State \Return $f_m(v) = \{(\mu_1, f_m^1), \dots, (\mu_k, f_m^k), \dots, (\mu_h, f_m^h)\}$
	\end{algorithmic}
	\label{BayesianUpdating}
\end{algorithm}

\noindent\textbf{Probability updating for the Laplace Mechanism.} 
In Equation \ref{PM1v}, the method for computing $P(\mathcal{A}_j(v) \mid v)$ is to uniformly discretize both the input and output domains into multiple buckets. This approach is relatively straightforward for PM, SW, and SR, as their input domains and output domains are bounded. However, the Laplace mechanism poses a challenge. Since the Laplace distribution is unbounded, its output domain is also unbounded. This prevents the use of a simple discretization approach, unlike the other mechanisms. In addition, the Laplace mechanism perturbs the original values by directly adding noise. This method of perturbation makes the calculation of the discretized perturbation probability less intuitive. 
To leverage the information provided by the Laplace mechanism, we need to develop a discretization approach for its unbounded output domain in order to determine the posterior distribution of the original value $v$.

Firstly, similar to other mechanisms, we discrete input domain into $h$ buckets, with the midpoint of each bucket denoted as $\{\mu_1,\dots,$\\$\mu_k,\dots,\mu_h\}$. Based on the different $\mu_k$ values, we can obtain the perturbed result $\mathcal{A}_{lap}(\mu_k)$ following these Laplace distributions:
\begin{equation}
	\setlength{\abovedisplayskip}{0pt} 
	\setlength{\belowdisplayskip}{0pt} 
	\rho(\mathcal{A}_{lap}(\mu_k) = x | \mu_k) = \frac{\epsilon_{lap}}{4} \exp\left(-\frac{\epsilon_{lap}|x-\mu_k|}{2}\right).
	\label{LAP}
\end{equation}

The perturbed value is $\mathcal{A}_{lap}(v)$, which is a constant. Our task is to calculate $P(\mathcal{A}_{lap}(v) \mid \mu_k)$ ($k \in {1, \dots, h}$) as shown in Equation \ref{Bucket}. To do so, we need to appropriately determine an output bucket $B'$. The selection of this $B'$ should satisfy three key requirements: (a) the bucket must be bounded to ensure practical computability; (b) it should accurately reflect the sensitivity of the output $\mathcal{A}_{lap}(v)$ to different input values $\mu_k$, thereby capturing the probabilistic variation characteristics of the Laplace distribution; and (c) for a given output $\mathcal{A}_{lap}(v)$, the width of the corresponding bucket should be consistent across different values of $\mu_k$, maintaining fairness in our probability estimations.

For requirement (a), we need to truncate the unbounded output domain into a bounded one. The probability of the Laplace distribution decreases with increasing distance from the axis of symmetry, so the information contained in the tails becomes less significant. This characteristic allows us to truncate the distribution without significant information loss, thereby enabling us to determine an appropriate $B'$. We define the output bucket $B'$ by the truncation distance $\tau$:
\begin{equation*}
	\setlength{\abovedisplayskip}{0pt} 
	\setlength{\belowdisplayskip}{0pt} 
	B' = \left[\mathcal{A}_{lap}(v) - \tau, \mathcal{A}_{lap}(v) + \tau\right],
\end{equation*}
where the truncation distance $\tau$ is defined as:
\begin{equation*}
	\setlength{\abovedisplayskip}{0pt} 
	\setlength{\belowdisplayskip}{0pt} 
	\tau = \sqrt{|\mathcal{A}_{lap}(v)|}
\end{equation*}
Since $\tau$ depends solely on $|\mathcal{A}_{lap}(v)|$, which is a constant, the length of the output bucket $B'$ remains consistent across different $\mu_k$ values. This choice of $\tau$ ensures that the bucket $B'$ satisfies requirements (b) and (c).

The probability $P(\mathcal{A}_{lap}(v)\in B'|\mu_k)$ is given by:
\begin{equation*}
	\setlength{\abovedisplayskip}{0pt} 
	\setlength{\belowdisplayskip}{0pt} 
	P(\mathcal{A}_{lap}(v)\in B'|\mu_k)=\int_{\mathcal{A}_{lap}(v) - \tau} ^{\mathcal{A}_{lap}(v) + \tau} 	\rho(\mathcal{A}_{lap}(\mu_k) = x | \mu_k) dx.
\end{equation*}

Finally, we normalize this value to ensure that the sum of all $P(\mathcal{A}_{lap}(v)|\mu_t)$ for $t \in \{1, \dots, h\}$ equals 1. And then we can obtain $P(\mathcal{A}_{lap}(v)|\mu_k)$ is
\begin{equation*}
	\setlength{\abovedisplayskip}{0pt} 
	\setlength{\belowdisplayskip}{0pt} 
	P(\mathcal{A}_{lap}(v)|\mu_k) = \frac{P(\mathcal{A}_{lap}(v)\in B'|\mu_k)}{\sum_{t=1}^{h} P(\mathcal{A}_{lap}(v)\in B'|\mu_t)}.
\end{equation*}

\subsection{User-level Weighted Averaging (UWA)}
After obtaining the posterior distribution for each user, we assign weights to different mechanisms based on this information. We propose the User-level Weighted Averaging (UWA) Algorithm to minimize variance under different settings. Before delving into the details of UWA, we need to explain how to calculate the variance based on the $f_m(v)=\{(\mu_1, f_m^1),\dots,(\mu_k, f_m^k),\dots,(\mu_h, f_m^h)\}$. The variance for $f_m(v)$ in terms of $\mathcal{A}_j$ is: 
\begin{equation}
	\setlength{\abovedisplayskip}{0pt} 
	\setlength{\belowdisplayskip}{0pt} 
	V_j(f_m(v))=\sum_{k=1}^{h}f_m^k var_j(\mu_k),
	\label{variancecalculation}
\end{equation}
where the variance $var_i(\mu_k)$ represents the output variance when the input value is $\mu_k$ and the perturbation mechanism $\mathcal{A}_i$ is applied. The methods for calculating variances of different mechanisms are detailed in Section \ref{preliminaries}.

\noindent\textbf{The procedure of UWA.} Algorithm \ref{UWA_algorithm} presents our user-level weighted averaging process. It takes the perturbation mechanisms, their corresponding privacy budgets, and perturbed results from $n$ users as input. For each user, the algorithm calculates posterior distributions using Algorithm \ref{BayesianUpdating} (line 1) and initializes weights (line 2). It then computes variances for each mechanism across the discretized values $\{\mu_1,\dots,\mu_h\}$, calculating $V_j(f_m(v_i))$ as the estimated variance for mechanism $\mathcal{A}_j$ (line 3). Based on $V_j(f_m(v_i))$, the algorithm determines optimal weights (line 4) and computes a weighted mean $v'_i$ for the $i$-th user(line 5). Finally, it aggregates the results by averaging the weighted means of all users (line 6).
\begin{algorithm}
	\caption{User-level weighted averaging algorithm}
	\begin{flushleft}
		{\bf Input:}
		Perturbation mechanisms $\{\mathcal{A}_1,\mathcal{A}_2,...,\mathcal{A}_m\}$, corresponding privacy budget $\{\epsilon_1,\epsilon_2,...,\epsilon_m\}$, perturbed result sets from $n$ users $\{U_1, U_2,..., U_n\}$.
		\\
		{\bf Output:}
		The aggregated mean $\hat{M}$
	\end{flushleft}
	\begin{algorithmic}[1]
		\\	\For {$i=1\ to\ n$}
		{ Calculate $f_m(v_i)=\{(\mu_1,f_m^1),\dots,(\mu_h,f_m^h)\}$ according to Algorithm \ref{BayesianUpdating}\\
			\quad Initialization weight: $w_t=\frac{1}{m},\   t=\{1,2,...,m\}$ \\ \quad
			\For{$j=1\ to\ m$} {\For{$t=1\ to\ h$}{
					Calculate $var_j(\mu_t)$ according to Section 2 
					
				}	 Calculate $V_j(f_m(v_i))$ according to Equation \ref{variancecalculation}}\\ \quad
			\For{$j=1\ to\ m$}{
				$w_j=[V_j(f_m(v_i))\sum_{t=1}^{h}\frac{1}{V_t(f_m(v_i))}]^{-1}$ }\\ \quad
			$v'_i=\sum_{j=1}^m w_j\mathcal{A}(v_i)$}
		\\  $\hat{M}=\sum_{i=1}^n\frac{1}{n}v'_i$
		\State \Return $\hat{M}$
	\end{algorithmic}
	\label{UWA_algorithm}
\end{algorithm}
\vspace{-0.1in}

For each original value $v$, we calculate a weighted mean $v'$. The weights are assigned to the perturbed values of $v$ to minimize variance, as shown in Theorem \ref{theoremofvar}.
\begin{theorem}
	\label{theoremofvar}
	Given the original value $v$ and corresponding variance $V_1(f_m(v)),\dots,V_m(f_m(v))$. The variance of $v'=\sum_{j=1}^m w_j\mathcal{A}_j(v)$ reaches the minimum, if the following formula holds:
	\begin{equation*}
		\setlength{\abovedisplayskip}{0pt} 
		\setlength{\belowdisplayskip}{0pt} 
		w_t=\frac{1}{V_t(f_m(v))\sum_{i=1}^{m}\frac{1}{V_i(f_m(v))}}, t\in\{1,\dots,m\}.
	\end{equation*}
	
	The minimal variance is:
	{\setlength{\abovedisplayskip}{0pt}
		\setlength{\belowdisplayskip}{0pt}
		\begin{equation*}
			\setlength{\abovedisplayskip}{0pt} 
			\setlength{\belowdisplayskip}{1pt} 
			Var(v')_{min}=[\sum_{t=1}^{m}\frac{1}{V_t(f_m(v))}]^{-1}.
	\end{equation*}}
\end{theorem}

\begin{proof}
	The variance of $v'$, which is a linear combination of $\mathcal{A}_t(v)$, can be written as:
	\begin{equation}
		\setlength{\abovedisplayskip}{1pt} 
		\setlength{\belowdisplayskip}{1pt} 
		\begin{split}
			Var(v')=Var(\sum_{t=1}^{m}w_{t}\mathcal{A}_t(v))=\sum_{t=1}^{m}w_{t}^{2}V_t(f_m(v)).
		\end{split}
		\label{variance_of_M}
	\end{equation}
	where $\sum{w_t}=1$.

	We regard the variance as a function of $w_t$, and the minimal variance is the extreme point of Equation \ref{variance_of_M}. By the Lagrangian method, we have:
	\begin{equation*}
		\setlength{\abovedisplayskip}{1pt} 
		\setlength{\belowdisplayskip}{1pt} 
		\mathcal{L}=\sum_{t=1}^{m}{w_t^{2}}V_t(f_m(v))+C_0(1-\sum_{t=1}^{h}w_t).
	\end{equation*}
	The first partial derivatives of $\mathcal{L}$ w.r.t. $w_t$ is $
	\frac{\partial \mathcal{L}}{\partial w_t}={2w_t}V_t(f_m(v))-C_0.$
	Let $\frac{\partial \mathcal{L}}{\partial w_t}=0$, then we have $wt=\frac{C_0}{2V_t(f_m(v))}$. Through the restriction $\sum_{t=1}^{m}w_t=1$, we figure out
	\begin{equation*}
		\setlength{\abovedisplayskip}{1pt} 
		\setlength{\belowdisplayskip}{1pt} 
		C_0=\frac{2}{\sum_{t=1}^{m}\frac{1}{V_t(f_m(v))}},\ w_t=\frac{1}{V_t(f_m(v))\sum_{i=1}^{m}\frac{1}{V_i(f_m(v))}}.
	\end{equation*}
	And the final minimal variance of $\tilde{M}$ is:
	\begin{equation*}
		\setlength{\abovedisplayskip}{1pt} 
		\setlength{\belowdisplayskip}{1pt} 
		{Var(v')}_{min}=[\sum_{t=1}^{m}\frac{1}{V_t(f_m(v))}]^{-1}.
	\end{equation*}
	\label{proof of cutoffpoint}
\end{proof}

Since users are mutually independent, the overall variance across all users will be minimized by summing the minimum variance attained for each individual perturbed value.
\begin{equation*}
	\setlength{\abovedisplayskip}{1pt} 
	\setlength{\belowdisplayskip}{1pt} 
	{Var(\{v'_1,\dots,v'_n\})}_{min}=\sum_{i=1}^n[\sum_{t=1}^{m}\frac{1}{V_t(f_m(v_i))}]^{-1}.
\end{equation*}
While UWA and UA perform similarly when services have comparable estimation performance, UWA demonstrates significant advantages in scenarios when different services exhibit varying estimation performance. In such cases, UA's uniform weighting leads to results that are less effective, whereas UWA maintains effectiveness through appropriate weight assignment.

\section{User-level Likelihood Estimation (ULE) for Distribution Estimation}
\label{Tech2}
Some mechanisms (e.g., the SR) are not designed for estimating distributions, and their perturbations can degrade a significant amount of effective information, leading to high inaccuracies if these perturbed results are used directly to aggregate the original distribution. However, we can make the most of all the perturbed results by treating these perturbed values for each user as a whole, because the perturbed results from the same user all stem from the same original value. By applying the Maximum Likelihood Estimation (MLE), we can better estimate the distribution of the entire user population (to clarify, this is distinct from Section 4.2, where we focus on estimating the posterior distribution of each individual user's value). Based on this method, we propose a general framework called User-level Likelihood Estimation (ULE), which leverages information from different services to provide more accurate distribution estimates.

\noindent\textbf{Likelihood estimation.} Recall that for each value $v_i$, its perturbed values from $m$ mechanisms are represented as $U_i = \{\mathcal{A}_1(v_i), \dots,$\\$ \mathcal{A}_m(v_i)\}$. For each service $j (j\in\{1,2,\dots,m\})$, its perturbation mechanism is $A_j$, and its collection is $C_j = \{\mathcal{A}_j(v_1), \dots,$ $\mathcal{A}_j(v_n)\}$. To adopt the MLE, we formulate the log-likelihood function of $F$ as follows:
\begin{equation}
	\label{F1}
	\setlength{\abovedisplayskip}{0pt}
	\setlength{\belowdisplayskip}{0pt}
	\begin{split}
		&L(D)=lnP[C_1,\dots,C_m|D]
	\end{split}
\end{equation}
where $D=\{d_1,\dots,d_l\}$ represents the original distribution over $l$ buckets for all users. Specifically, we discretize the input domain into $l$ buckets, denoted as $\{B_1,\dots,B_k,\dots,B_l\}$, where $d_k$ represents the proportion of data falling within bucket $B_k$.

We consider distribution from the perspective of users, and the Equation \ref{F1} can be written as:
\begin{equation}
	\label{F2}
	\setlength{\abovedisplayskip}{0pt}
	\setlength{\belowdisplayskip}{0pt}
	\begin{split}
		L(D)&=lnP[U_1,\dots,U_n|D].\\
	\end{split}
\end{equation}

\textbf{Expectation-Maximization (EM) algorithm.} Our goal is to determine the parameter $D$ that maximizes the likelihood function $L(D)$. To achieve this, we employ the Expectation-Maximization (EM) algorithm, inspired by its successful application in distribution recovery estimation for CFO and SW ~\cite{2016Building, ren2018textsf, li2020estimating}. The EM algorithm converges to the MLE of the original distribution $D$ according to the following theorem.

\begin{theorem}
	The EM algorithm converges to the maximum-likelihood estimator of the true frequencies D.
\end{theorem}

\begin{proof}
	Equation \ref{F2} can be written as:
	\begin{equation}
		\setlength{\abovedisplayskip}{0pt} 
		\setlength{\belowdisplayskip}{0pt} 
		\label{Object function3}
		\begin{split}
			L(D)=\sum_{i=1}^{n}ln(\sum_{k=1}^{l}d_k P[\mathcal{A}_1(v_i),\dots,\mathcal{A}_m(v_i)|v_i\in B_k]).
		\end{split}
	\end{equation}
	
	Since the perturbations of different mechanisms are independent of each other, we have
	\begin{equation}
		\setlength{\abovedisplayskip}{0pt} 
		\setlength{\belowdisplayskip}{0pt} 
		\label{Object function4}
		\begin{split}
			& P[\mathcal{A}_1(v_i),\dots,\mathcal{A}_j(v_i),\dots,\mathcal{A}_m(v_i)|v_i\in B_k]=\\&P[\mathcal{A}_1(v_i)|v_i\in B_k]\dots P[\mathcal{A}_j(v_i)|v_i\in B_k]\dots P[\mathcal{A}_m(v_i)|v_i\in B_k],
		\end{split}
	\end{equation}
	where $P[\mathcal{A}_j(v_i)|v_i\in B_k]$ is a constant determined by perturbation mechanism, and thus $L(D)$ is a concave function. So the EM algorithm converges to the MLE of the original distribution $D$.
\end{proof}

To compute $P[\mathcal{A}_j(v_i) | v_i \in B_k]$, we also discretize the output domain of $\mathcal{A}_j$ into $h_j$ uniform buckets. For $\mathcal{A}_j$, its output domain is discretized into buckets $H_j[1], H_j[2], \dots, H_j[r], \dots, H_j[h_j]$. We need to determine which bucket $\mathcal{A}_j(v_i)$ belongs to and then calculate its probability. For instance, consider $P[\mathcal{A}_j(v)\in H_j[r]|v\in B_k]$, which represents the probability that the perturbed value $\mathcal{A}_j(v)$ falls into bucket $H_j[r]$, given that the original value $v$ falls into bucket $B_k$.

In Equation \ref{Object function4}, for $v_i$, the perturbation results will fall into different buckets. To simplify the representation and facilitate the computation, we can represent the combination of these buckets as a bucket vector. Let's define the $r$-th bucket vector $W_r$ as a set containing $m$ elements: $W_r = \{W_r[1], W_r[2], \dots, W_r[m]\}$. Each element $W_r[j]$ represents one of the discretized buckets $H_j[1], \dots, H_j[r], \dots, $\\$ H_j[h_j]$ from the output domain of $\mathcal{A}_j$. When we say $\{\mathcal{A}_1(v),\dots, $\\$\mathcal{A}_m(v)\} \in W_r$, it means $\mathcal{A}_1(v) \in W_r[1], \mathcal{A}_2(v) \in W_r[2], \dots, \mathcal{A}m(v) $\\$\in W_r[m]$. Consequently, we can reformulate Equation \ref{Object function4} as follows:
\begin{equation}
	\setlength{\abovedisplayskip}{0pt} 
	\setlength{\belowdisplayskip}{0pt} 
	\label{Object function6}
	\begin{split}
		L(F)=\sum_{i=1}^{n}ln(\sum_{k=1}^{l}d_k P[\{\mathcal{A}_1(v_i),\dots,\mathcal{A}_m(v_i)\}\in W_r|v_i\in B_k]).
	\end{split}
\end{equation}
The number of possible bucket vectors, denoted as $l'$, is given by $l' = \min\{n, \prod_{j=1}^{m} h_j\}$. Then, we employ the EM algorithm to determine the MLE of Equation \ref{Object function6}, following the methodology described by Li et al. \cite{li2020estimating}. \textbf{E step} computes the conditional probability distribution of the latent data $P_t$ given the current estimate of the model parameters:
\begin{equation}
	\begin{split}
		\setlength{\abovedisplayskip}{0pt} 
		\setlength{\belowdisplayskip}{0pt} 
		&P_t = \hat{d}_t \sum_{k=1}^{l'} n_k \frac{P [\{ \mathcal{A}_1(v),\dots, \mathcal{A}_m(v)\}\in W_k | \mathbf{v} \in {B}_t, {\hat{D}}]}{P [\{ \mathcal{A}_1(v),\dots, \mathcal{A}_m(v)\}\in W_k | {\hat{D}}]}=
		\\& \hat{d}_t  \! \sum_{k=1}^{l'}  \! \frac{n_k P [\mathcal{A}_1(v) \in W_k[1] | \mathbf{v} \in {B}_t]  \! \dots  \! P [\mathcal{A}_m(v) \in W_k[m] | \mathbf{v} \in {B}_t]}{ \!\sum_{k=1}^l \! {P [\mathcal{A}_1(v)  \! \in \! W_k[1] | \mathbf{v}  \!\in \! {B}_t]  \!\dots \! P [\mathcal{A}_m(v)  \!\in \! W_k[m] | \mathbf{v}  \!\in \! {B}_t]} \hat{d}_k},
	\end{split}
	\label{estep}
\end{equation}
where $n_k$ denotes the count of perturbed values corresponding to the bucket vector $W_k$.

\textbf{M step} calculates parameters $\hat{d}_t$($t\in\{1,\dots,l\}$) that maximize the expected log-likelihood obtained in the E step:
\begin{equation}
	\setlength{\abovedisplayskip}{0pt} 
	\setlength{\belowdisplayskip}{0pt} 
	\hat{d}_t = \frac{P_t}{\sum_{k=1}^l P_k}.
	\label{mstep}
\end{equation}

Building upon this approach, we propose the ULE algorithm.
\noindent\textbf{The procedure of ULE.}
The ULE algorithm reconstructs the distribution as illustrated in Algorithm \ref{alg:post-proc-em}. The algorithm begins by assigning non-zero initial values to $\hat{D}$, ensuring $\sum_{t=1}^l \hat{d_t}=1$ (line 1). It then executes the EM algorithm, alternating between the Expectation (E) and Maximization (M) steps (lines 3-4). In the E-step, the algorithm evaluates the log-likelihood expectation using the current estimate $\hat{D}$ and the observed counts $n_k$, where $n_k$ denotes the count of perturbed values corresponding to the bucket vector $W_k$. The M-step calculates a new $\hat{D}$ that maximizes the expected likelihood, which serves as input for the next E-step (line 4). The algorithm terminates when the results converge and returns the estimated frequency histogram $\hat{D}$ once the convergence condition is met (line 5).
\begin{algorithm}
	\caption{Procedure for ULE}
	\label{alg:post-proc-em}
	\begin{flushleft}
		{\bf Input:}
		Perturbation mechanisms $\{\mathcal{A}_1,\mathcal{A}_2,...,\mathcal{A}_m\}$, corresponding privacy budget $\{\epsilon_1,\epsilon_2,...,\epsilon_m\}$, perturbed result sets from $n$ users $\{U_1, U_2,..., U_n\}$, number of discrete buckets for the input domain $l$.
		\\
		{\bf Output:}
		The aggregated distribution $\hat{D}$
	\end{flushleft}
	\begin{algorithmic}[1]
		\\
		$\hat{d_t}=\frac{1}{l}, t\in\{1,\dots,l\}$\\
		\While{not converge}
		{\textbf{E-step:}\\
			\For{ $t \in \{1, \dots, l\}$}
			{ 
				Calculate $P_t$ according to Equation \ref{estep}
			}
			\textbf{M-step:}\\	
			\For{$t \in \{1, \dots, l\}$}
			{ 	Calculate $\hat{d}_t$ according to Equation \ref{mstep}}}
		\State \textbf{Return} $\hat{D}=\{\hat{d}_1,\dots,\hat{d}_l\}$
	\end{algorithmic}
\end{algorithm}
\vspace{-0.1in}
To summarize, the performance advantage of ULE over any single service primarily depends on the informativeness of data from other services, which is determined by four key factors: the number of services, privacy budget settings, mechanism types, and user base size.

\section{Experimental Results}
\label{exp}
\subsection{Experimental Setup}
We conduct experiments on a workstation equipped with an AMD Ryzen 7 2700X eight-core processor, 64 GB of RAM, an NVIDIA GeForce RTX 3090 GPU, and Windows 10 operating system. All experiments are implemented in Python 3.12.0. The experimental code and datasets are available online\footnote{https://github.com/RONGDUGithub/Distribution\_ULE\_1}. 

\subsubsection{Experiment Design} 
Each of $m$ services may employ the same or different perturbation mechanisms for estimating statistical results. We will demonstrate that our proposed methods achieve better accuracy than any individual service, while preventing further privacy leakage of user data. For mean estimation, we compare our proposed methods UA and UWA against existing mechanisms such as Laplace, SW, PM, and SR. For distribution estimation, we compare our proposed method, ULE, against SW, PM, and SR mechanisms. We exclude the Laplace mechanism from our distribution estimation framework due to its limitations in distribution estimation tasks. When discretizing the continuous perturbed data, truncation to a feasible range is necessary, leading to significant information loss that distorts the distribution estimate. In contrast, the SR, PM, and SW mechanisms can be discretized without truncation, avoiding additional information loss. The stopping criterion for EM in our paper is when $|L(\hat{F}^{(t+1)}) - L(\hat{F}^{(t)})| < \epsilon$, which is the same as setting in paper~\cite{li2020estimating}.

\subsubsection{Datasets} 
Our analysis employs four datasets --- two synthetic and two real-world numerical sets. The first synthetic dataset, $\textbf{Beta(2,5)}$, is drawn from the Beta distribution~\cite{johnson1995continuous} where the shape parameters $\alpha = 2$ and $\beta = 5$ determine the skewness and concentration of the distribution. $\textbf{Beta+Sin}$ is constructed by superimposing sinusoidal noise onto the $\textbf{Beta(2,5)}$ distribution to introduce periodic fluctuations in the probability density. Each synthetic dataset comprises 1,000,000 samples uniformly distributed within the interval $[0,1]$. For the real-world data analysis, we utilize the $\textbf{Taxi}$ dataset~\cite{taxi2018}, which consists of pick-up times extracted from the New York Taxi data for January 2018. This dataset encompasses 1,048,575 integer values, each representing the number of seconds within 24 hours (range: 0 to 86,340). The second real-world dataset, $\textbf{Retirement}$~\cite{retirement2013}, is sourced from the San Francisco city employee retirement plans, providing comprehensive data on salary and benefits disbursements since the fiscal year 2013. Our analysis specifically focuses on the total compensation data, incorporating a subset of 606,507 entries within the range of 10,000 to 60,000.

\subsubsection{Utility Metrics} \hfill

\noindent\textbf{MSE.} We use Mean Squared Error ($\text{MSE}$)~\cite{lehmann2006theory} to evaluate the accuracy of mean estimation. MSE is a commonly used metric to measure the average squared difference between predicted values $\hat{y}_i$ and the ground truth $y_i$ in regression tasks. Formally, 
\begin{equation*}
	\setlength{\abovedisplayskip}{0pt} 
	\setlength{\belowdisplayskip}{0pt} 
	MSE = \frac{1}{t} \sum_{i=1}^{t} (\hat{y}_i - y_i)^2.
\end{equation*}
where $t$ is the number of observations and we set $t=50$ (i.e., $50$ trials) in our experiments. 

\noindent\textbf{Jensen-Shannon divergence.}
The Jensen-Shannon divergence (JS divergence) \cite{dagan1997similarity}, denoted as $d_{\text{JS}}$, is to measure the similarity between two probability distributions. 
Given two probability distributions $P$ and $Q$, the Jensen-Shannon divergence is calculated as:
{\setlength{\abovedisplayskip}{0pt}
	\setlength{\belowdisplayskip}{0pt}\begin{equation*}
		\setlength{\abovedisplayskip}{0pt} 
		\setlength{\belowdisplayskip}{0pt} 
		d_{\text{JS}}(P, Q) = \sqrt{\frac{1}{2} D_{\text{KL}}(P \parallel M) + \frac{1}{2} D_{\text{KL}}(Q \parallel M)},
\end{equation*}}
where $M = \frac{1}{2}(P + Q)$ is the average distribution, and $D_{\text{KL}}$ represents the Kullback-Leibler divergence \cite{manning1999foundations} defined as:
\begin{equation*}
	\setlength{\abovedisplayskip}{0pt} 
	\setlength{\belowdisplayskip}{0pt} 
	D_{\text{KL}}(P \parallel Q) = \sum_{i} P(i) \log. \frac{P(i)}{Q(i)}
\end{equation*}

In our experiments, we employ the JS divergence to evaluate the performance of the distribution estimation. 
A smaller $d_{\text{JS}}$ indicates better performance, with a value of 0 representing identical distributions.

\subsection{Multiple Services Statistics}
In this set of experiments, we evaluate mean estimation across four heterogeneous services, each using a different mechanism (SR, Laplace, PM, and SW), with UA and UWA utilizing these services' information. For distribution estimation, we evaluate three services using three different mechanisms (SR, PM, and SW), with ULE utilizing their information. We compare the statistical accuracy of aggregated results against individual service estimates throughout all experiments.

\begin{figure*}[th]
	\centering
	{
		\begin{minipage}{6cm}
			\centering
			\includegraphics[scale=0.75]{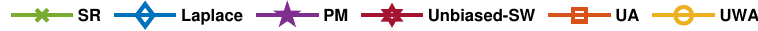}
		\end{minipage}
	}
	\hspace{1.6in}
	{
		\begin{minipage}{6cm}
			\centering
			\includegraphics[scale=0.75]{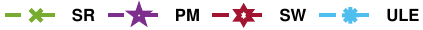}
		\end{minipage}
	}
	\\
	\vspace{-0.1in}
	\centering
	\subfigure[\textbf{Beta25}, mean.]{
		\begin{minipage}[t]{0.24\linewidth}
			\centering
			\includegraphics[width=1\textwidth]{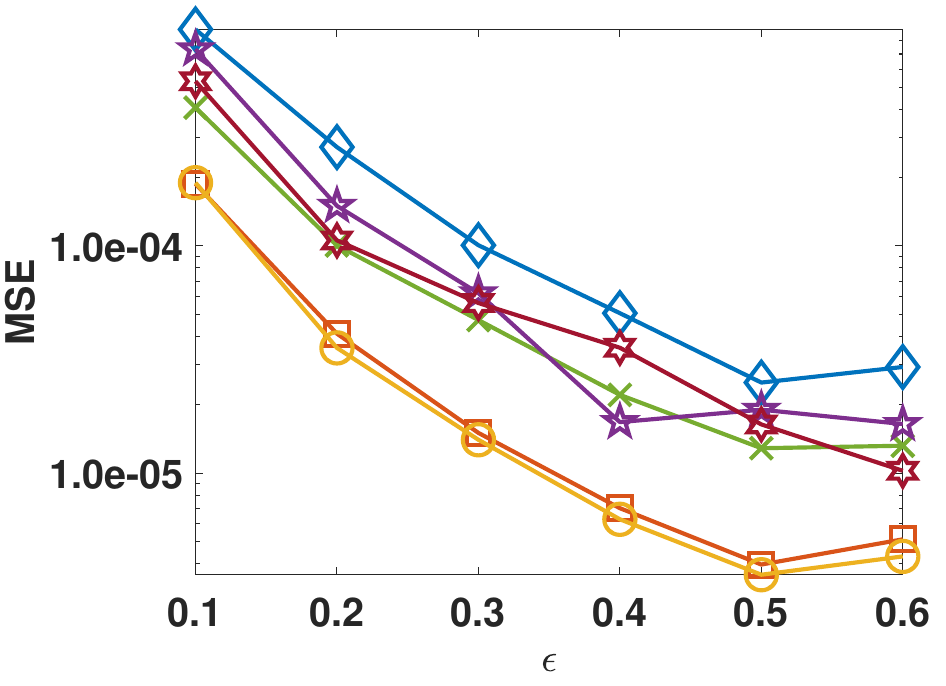}
		\end{minipage}%
	}%
	\subfigure[\textbf{Beta+Sin}, mean.]{
		\begin{minipage}[t]{0.24\linewidth}
			\centering
			\includegraphics[width=1\textwidth]{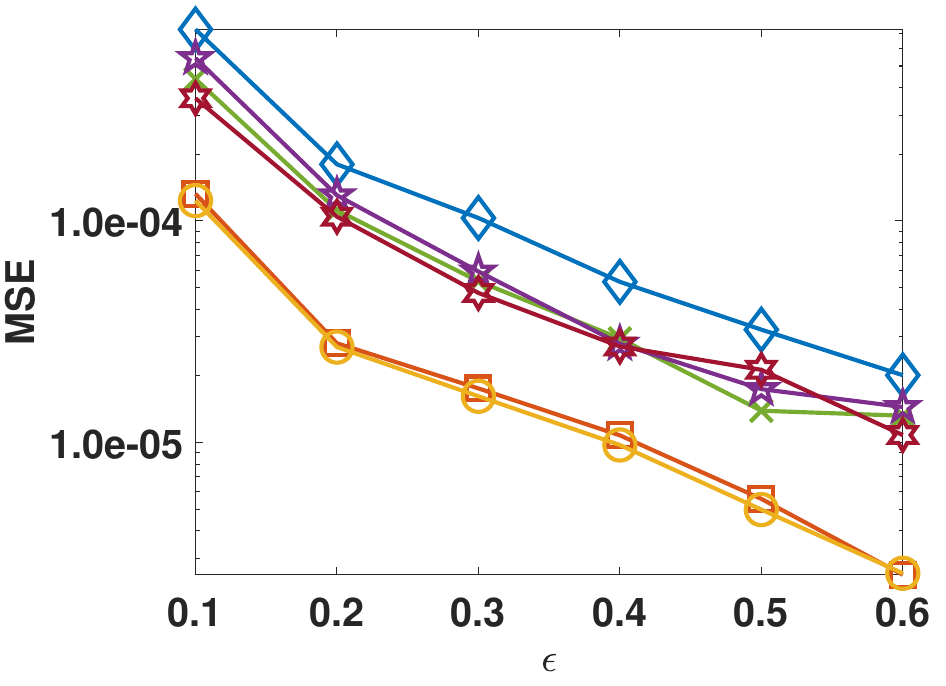}
		\end{minipage}%
	}%
	\subfigure[\textbf{Taxi, mean}.]{
		\begin{minipage}[t]{0.24\linewidth}
			\centering
			\includegraphics[width=1\textwidth]{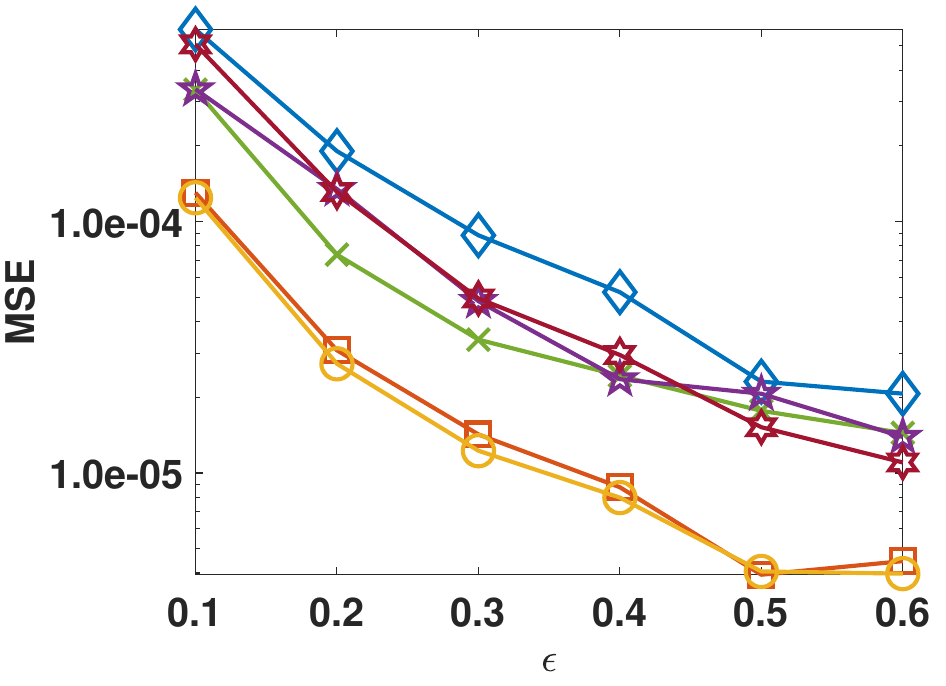}
		\end{minipage}
	}%
	\subfigure[\textbf{Retirement} , mean.]{
		\begin{minipage}[t]{0.24\linewidth}
			\centering
			\includegraphics[width=1\textwidth]{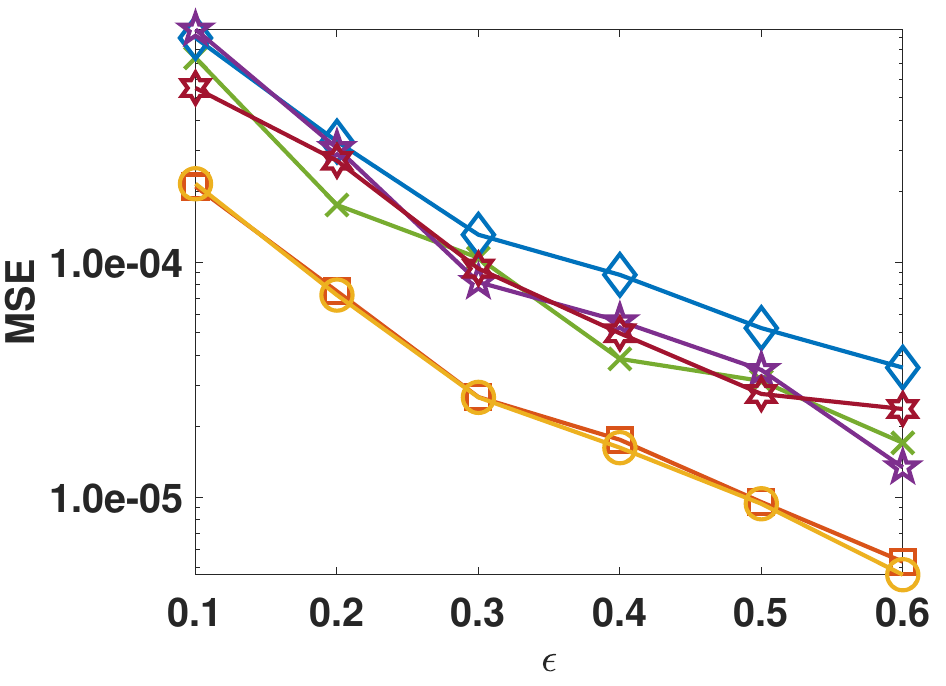}
		\end{minipage}
	}%
	\\
	\vspace{-0.15in}
	\centering
	\subfigure[\textbf{Beta25}, distribution.]{
		\begin{minipage}[t]{0.24\linewidth}
			\centering
			\includegraphics[width=1\textwidth]{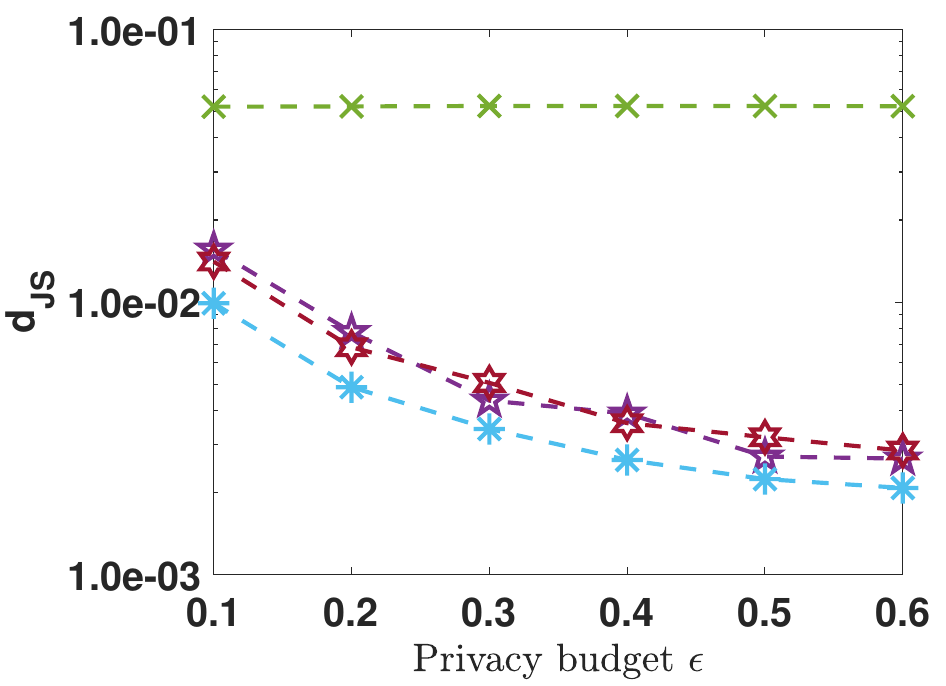}
		\end{minipage}%
	}%
	\subfigure[\textbf{Beta+Sin}, distribution.]{
		\begin{minipage}[t]{0.24\linewidth}
			\centering
			\includegraphics[width=1\textwidth]{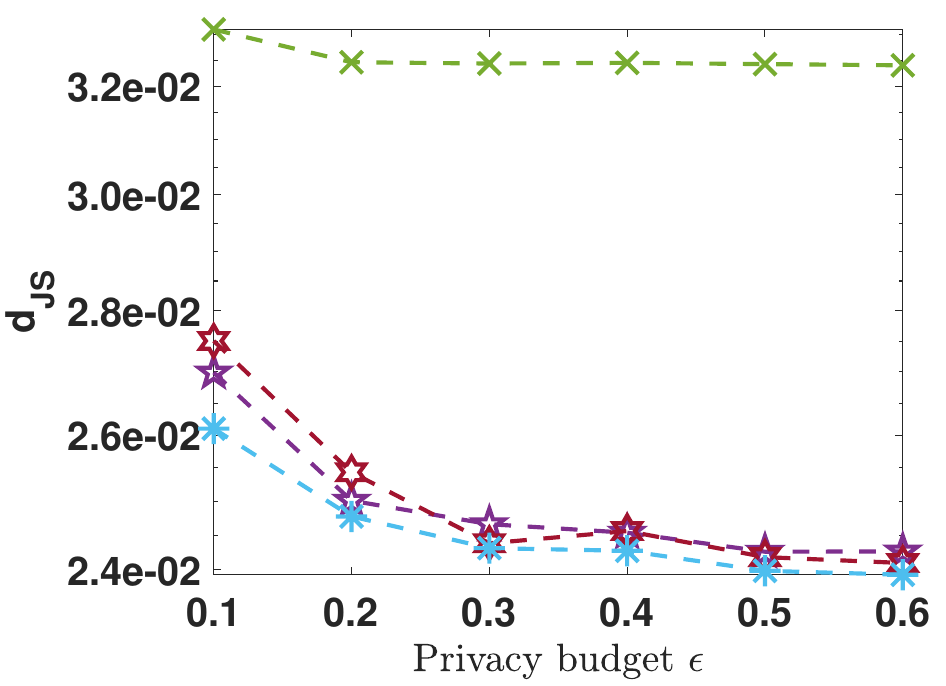}
		\end{minipage}%
	}%
	\subfigure[\textbf{Taxi}, distribution.]{
		\begin{minipage}[t]{0.24\linewidth}
			\centering
			\includegraphics[width=1\textwidth]{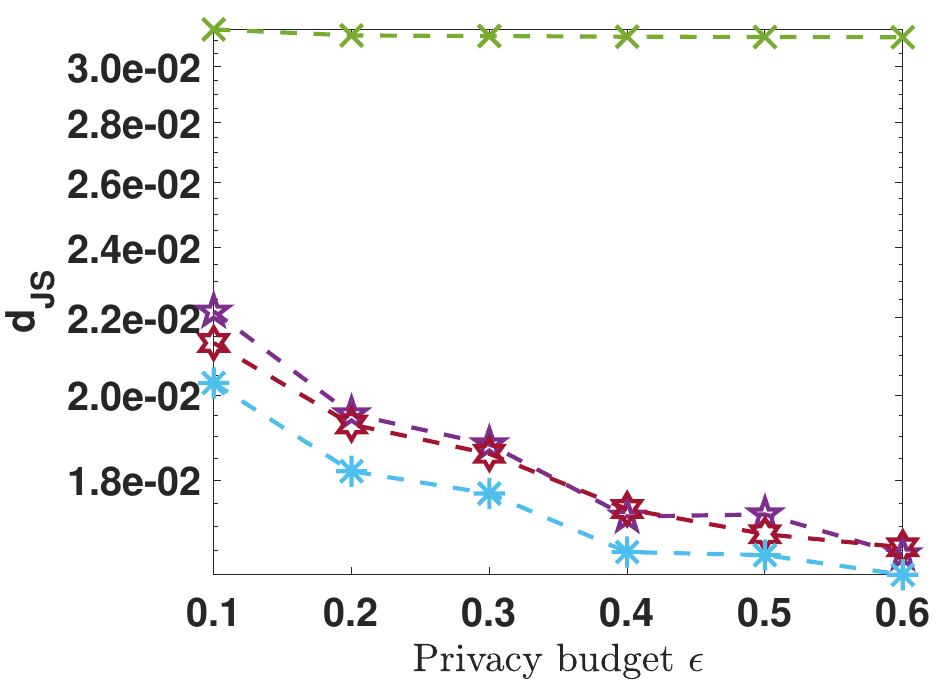}
		\end{minipage}
	}%
	\subfigure[\textbf{Retirement}, distribution.]{
		\begin{minipage}[t]{0.24\linewidth}
			\centering
			\includegraphics[width=1\textwidth]{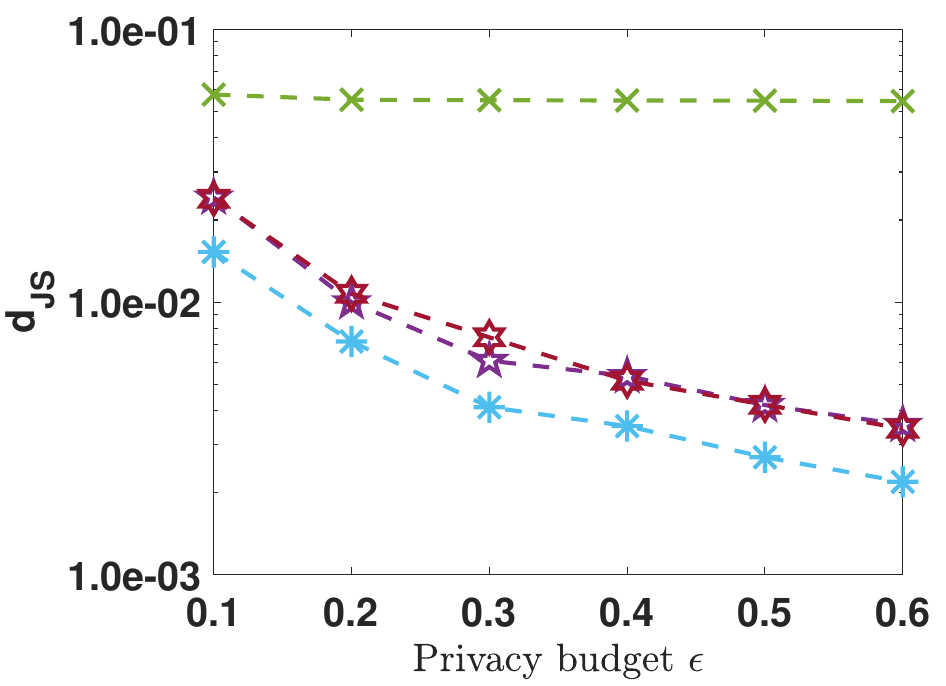}
		\end{minipage}
	}%
	\vspace{-0.05in}
	\caption{Comparison of MSE and JS divergence across different privacy budgets ($\epsilon$): mean estimation (a)-(d) and distribution estimation (e)-(h)}
	\label{Exp1}
	\vspace{-0.15in}
\end{figure*}

\noindent\textbf{Results with Same Privacy Budgets.} Figures \ref{Exp1}(a)-(d) show the MSE of mean estimation results with privacy budget $\epsilon$ increasing from 0.1 to 0.6. As $\epsilon$ increases, the MSE of all methods decreases. Our proposed methods, UA and UWA, demonstrate significant performance improvements of 53.3\%-85.0\% compared to the best results from any single mechanism. Furthermore, UWA consistently outperforms UA in most scenarios, with improvements typically ranging from 1\% to 15\% across different privacy budgets.

Figures \ref{Exp1}(e)-(h) present the distribution estimation results with varying $\epsilon$. The JS divergence of all methods decreases correspondingly with increasing $\epsilon$. SR exhibits the worst performance due to significant information loss from discretizing. PM and SW mechanisms achieve better performance than SR through their ability to preserve the original distribution shape, maintaining high probability density around the true values. By leveraging information from all available mechanisms, our proposed ULE consistently achieves superior performance across various datasets compared to any single service, demonstrating improvements of 21.1\%-81.0\% over SR, 3.3\%-36.9\% over PM, and 4.8\%-36.10\% over SW.

\noindent\textbf{Results with Different Privacy Budgets.} 
In Figures \ref{Exp2}(a)-(d), we investigate the impact of varying privacy budgets across mechanisms using four configurations (A1-A4) on mean estimation. The privacy budgets [SR, Laplace, PM, SW] are set to A1[0.1, 0.2, 0.3, 0.4], A2[0.2, 0.4, 0.1, 0.3], A3[0.3, 0.1, 0.4, 0.2], and A4[0.4, 0.3, 0.2, 0.1]. Each mechanism's performance directly correlates with its privacy budget, as demonstrated by the superior performance of mechanisms with a 0.4 budget (PM in A3 and SR in A4). UA's uniform weighting scheme gives excessive influence to highly perturbed results, leading to performance worse than some single mechanisms. In contrast, UWA leverages Bayesian approaches to determine appropriate weights, thus consistently achieving superior performance, surpassing the best single mechanism by 11.51\% to 72.41\%, with an average improvement exceeding 40\%.

Figures \ref{Exp2}(e)-(h) present the distribution estimation results under four privacy budget configurations [SR, PM, SW]: C1[0.1, 0.2, 0.3], C2[0.2, 0.3, 0.1], C3[0.3, 0.1, 0.2], and C4[0.2, 0.1, 0.3]. SR consistently exhibits a higher MSE regardless of its budget allocation. Between PM and SW, the mechanism with budget 0.3 achieves better results, as shown by SW in C1/C4 and PM in C2. Our proposed ULE consistently outperforms existing individual mechanisms across all four datasets, achieving improvements of 23.65\%-92.56\% over SR, 0.15\%-73.42\% over PM, and 0.02\%-73.32\% over SW.

ULE achieves only modest improvements over SW and PM in configurations C1/C4 and C2 respectively. This limited enhancement is primarily due to two factors: the high estimation accuracy already achieved by SW (C1/C4) and PM (C2) with maximum privacy budgets, and the constrained information available from other services due to smaller budgets or mechanism limitations. The SR mechanism, in particular, provides minimal contribution owing to significant information loss during perturbation.

\begin{figure*}[th]
	\hspace{0.2in}
	{
		\begin{minipage}{6cm}
			\centering
			\includegraphics[scale=0.75]{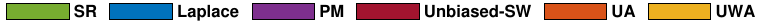}
		\end{minipage}
	}
	\hspace{1.6in}
	{
		\begin{minipage}{6cm}
			\centering
			\includegraphics[scale=0.75]{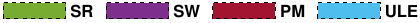}
		\end{minipage}
	}
	\\
	\vspace{-0.05in}
		\centering
		\subfigure[\textbf{Beta25}, mean.]{
			\begin{minipage}[t]{0.24\linewidth}
				\centering
				\includegraphics[width=1\textwidth]{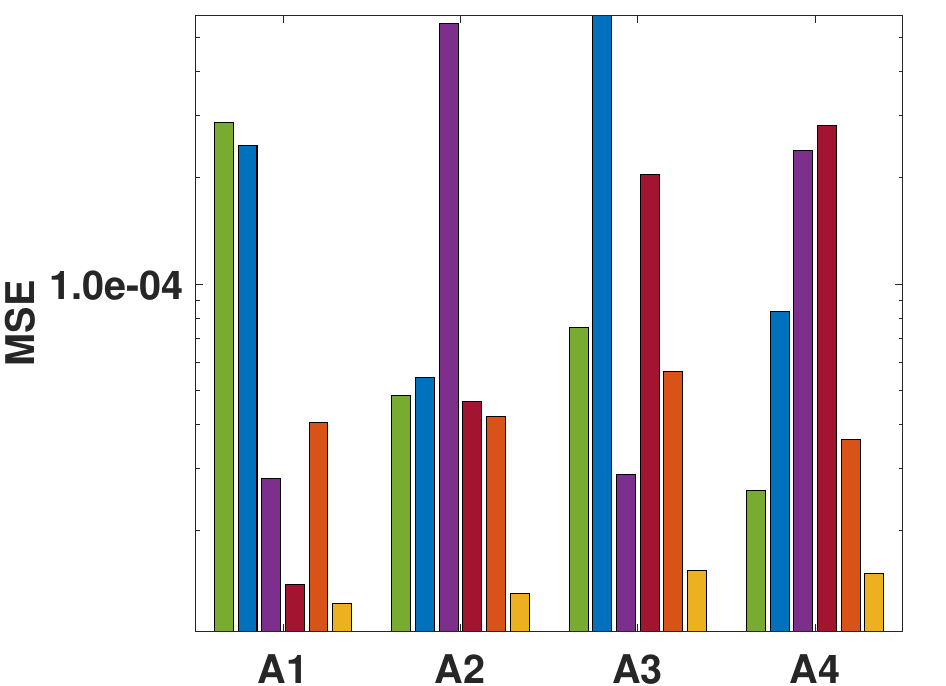}
			\end{minipage}%
		}%
		\subfigure[\textbf{Beta+Sin}, mean.]{
			\begin{minipage}[t]{0.24\linewidth}
				\centering
				\includegraphics[width=1\textwidth]{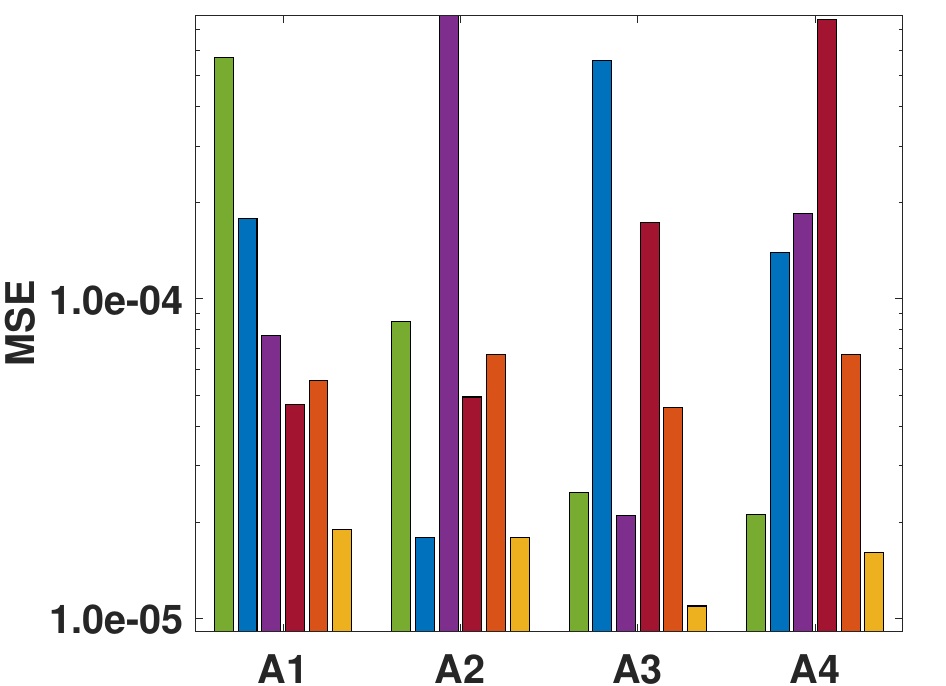}
			\end{minipage}%
		}%
		\subfigure[\textbf{Taxi}, mean.]{
			\begin{minipage}[t]{0.24\linewidth}
				\centering
				\includegraphics[width=1\textwidth]{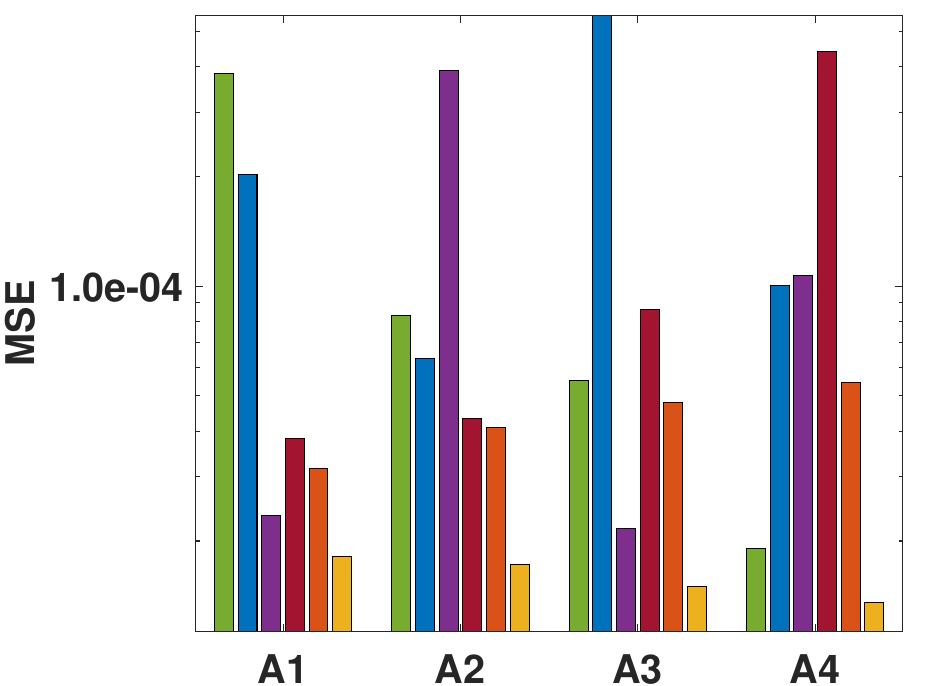}
			\end{minipage}
		}%
		\subfigure[\textbf{Retirement}, mean.]{
			\begin{minipage}[t]{0.24\linewidth}
				\centering
				\includegraphics[width=1\textwidth]{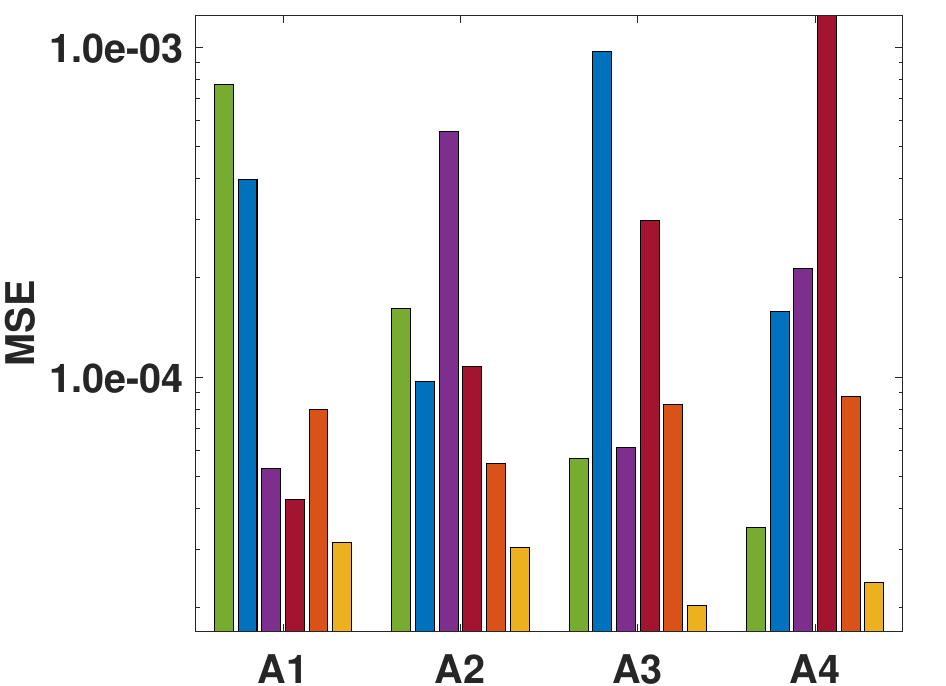}
			\end{minipage}
		}%
		\vspace{-0.15in}
		\\
		\centering
		\subfigure[\textbf{Beta25}, distribution.]{
			\begin{minipage}[t]{0.24\linewidth}
				\centering
				\includegraphics[width=1\textwidth]{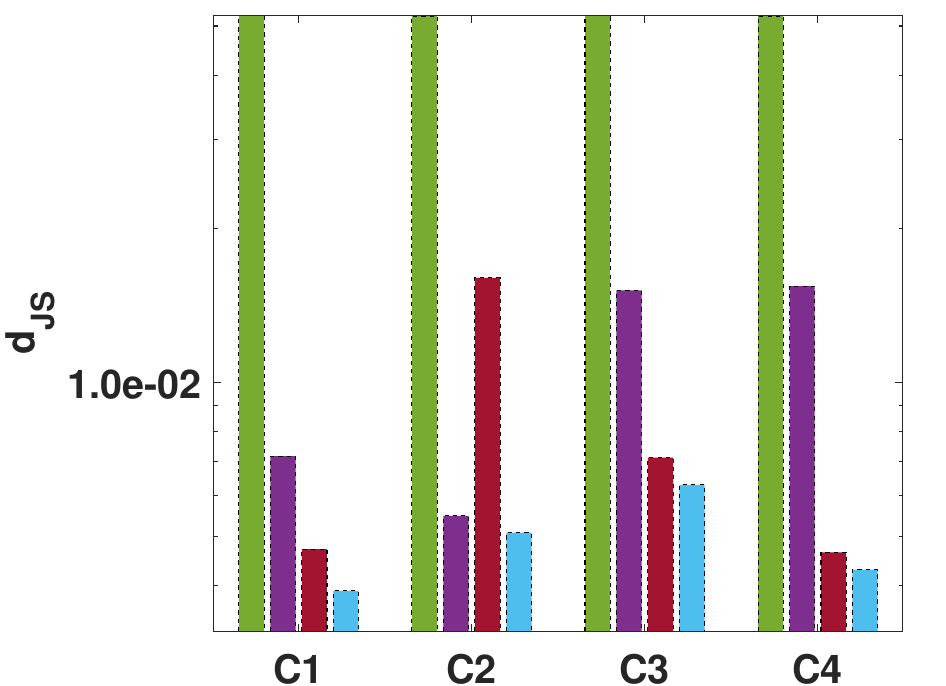}
			\end{minipage}%
		}%
		\subfigure[\textbf{Beta+Sin}, distribution.]{
			\begin{minipage}[t]{0.24\linewidth}
				\centering
				\includegraphics[width=1\textwidth]{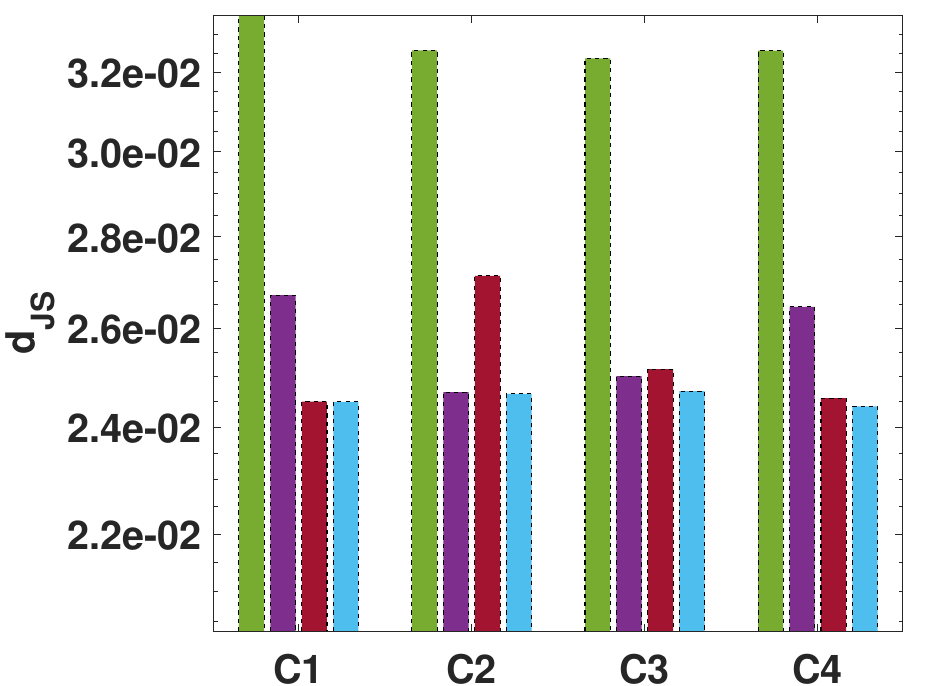}
			\end{minipage}%
		}%
		\subfigure[\textbf{Taxi}, distribution.]{
			\begin{minipage}[t]{0.24\linewidth}
				\centering
				\includegraphics[width=1\textwidth]{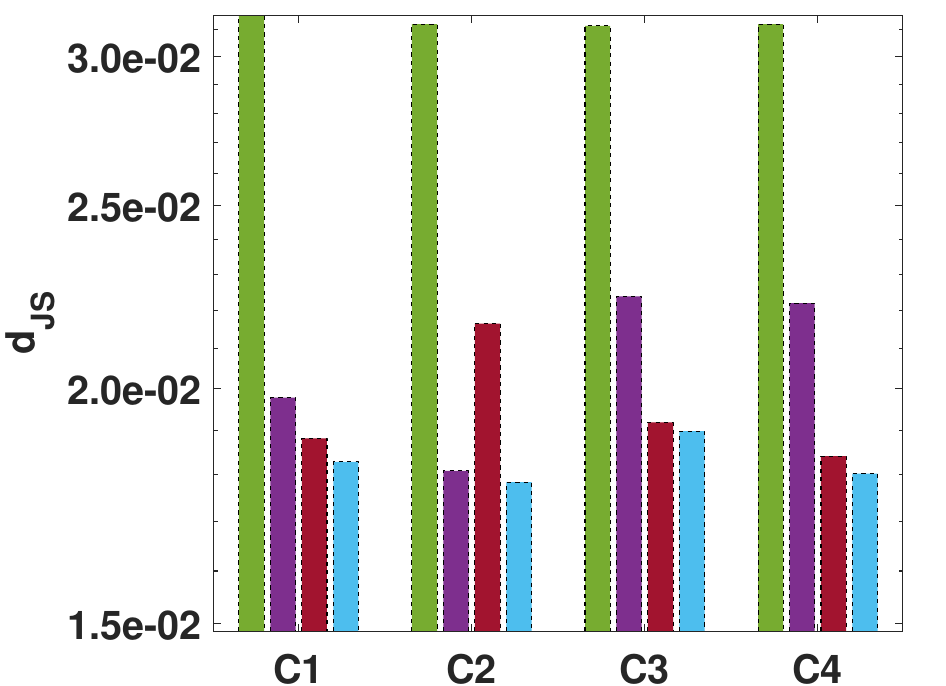}
			\end{minipage}
		}%
		\subfigure[\textbf{Retirement}, distribution.]{
			\begin{minipage}[t]{0.24\linewidth}
				\centering
				\includegraphics[width=1\textwidth]{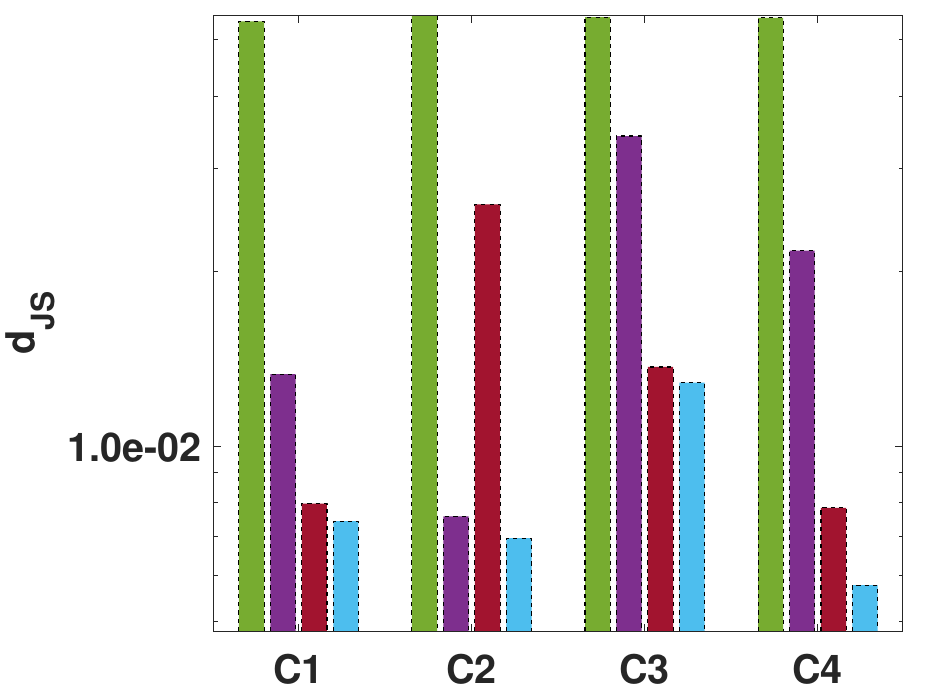}
			\end{minipage}
		}%
		\caption{Comparison of MSE and JS divergence across different $\epsilon$ configurations: mean estimation (a)-(d) and distribution estimation (e)-(h)}
		\label{Exp2}
		\vspace{-0.1in}
	\end{figure*}

	\subsection{Robustness Verification}
	To evaluate the robustness of our methods, we conduct experiments under four distinct scenarios: (1) implementation with a single perturbation mechanism for all services, (2) application to high-dimensional datasets, (3) performance with varying numbers of services, and (4) scalability across different user base sizes.
	
	\noindent\textbf{Impact of Same Mechanism.}
	Figures \ref{Exp3}(a)-(b) evaluate the scenario where four services employ identical mechanisms with privacy budgets ranging from $0.1$ to $0.4$. For mean estimation in Figure \ref{Exp3}(a), while all mechanisms improve with increasing $\epsilon$, UA performs worse than the single service with privacy budget $\epsilon=0.4$. This performance degradation stems from its simple averaging strategy across different privacy budgets, making it susceptible to services with poor performance. In contrast, UWA maintains the lowest MSE across all settings, demonstrating superior estimation accuracy regardless of mechanism type. Figure \ref{Exp3}(b) displays the JS divergence for distribution estimation, where three services utilize the identical mechanism with privacy budgets ranging from $0.1$ to $0.3$. ULE consistently achieves the lowest JS divergence across all settings. 
	
	\noindent\textbf{Applicability to High-Dimensional Data.}
	To demonstrate the effectiveness of our methods on high-dimensional data, we utilize the \textbf{Taobao}\footnote{https://tianchi.aliyun.com/dataset/56}, which comprises advertising click records from 1,172,556 users over three days. Each record captures the category of the last click in 10-minute intervals. After normalizing the 10,401 category values to [-1, 1] and extracting 60-timestamp sequences, we evaluate our methods' performance with dimensionality varying from 10 to 60, as shown in Figures \ref{Exp3}(c)-(d). To address the high dimensionality challenge, we adopt sampling strategy~\cite{erlingsson2014rappor} that partitions users instead of splitting the privacy budget across dimensions. Although performance degrades with increasing dimensionality due to fewer users per dimension, our UA and UWA methods maintain superior performance in mean estimation compared to individual services (Figure \ref{Exp3}(c)), with ULE showing better distribution estimation performance (Figure \ref{Exp3}(d)).

	\begin{figure*}[th]
		\hspace{-0.9in}
		{
			\begin{minipage}{6cm}
				\centering
				\includegraphics[scale=0.7]{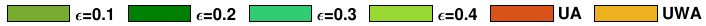}
			\end{minipage}
		}
		\hspace{1.0in}
		{
			\begin{minipage}{6cm}
				\centering
				\includegraphics[scale=0.7]{legend_G1_mean.pdf}
			\end{minipage}
		}
		\\
		\hspace{0.2in}
		\centering
		{
			\begin{minipage}{6cm}
				\centering
				\includegraphics[scale=0.7]{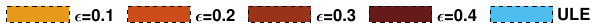}
			\end{minipage}
		}
		\hspace{1.6in}
		{
			\begin{minipage}{6cm}
				\centering
				\includegraphics[scale=0.7]{legend_G1_distribution.pdf}
			\end{minipage}
		}\\
		\vspace{-0.05in}
		\centering
		\subfigure[\textbf{Beta+Sin}, $\epsilon$=0.1, mean.]{
			\begin{minipage}[t]{0.24\linewidth}
				\centering
				\includegraphics[width=1\textwidth]{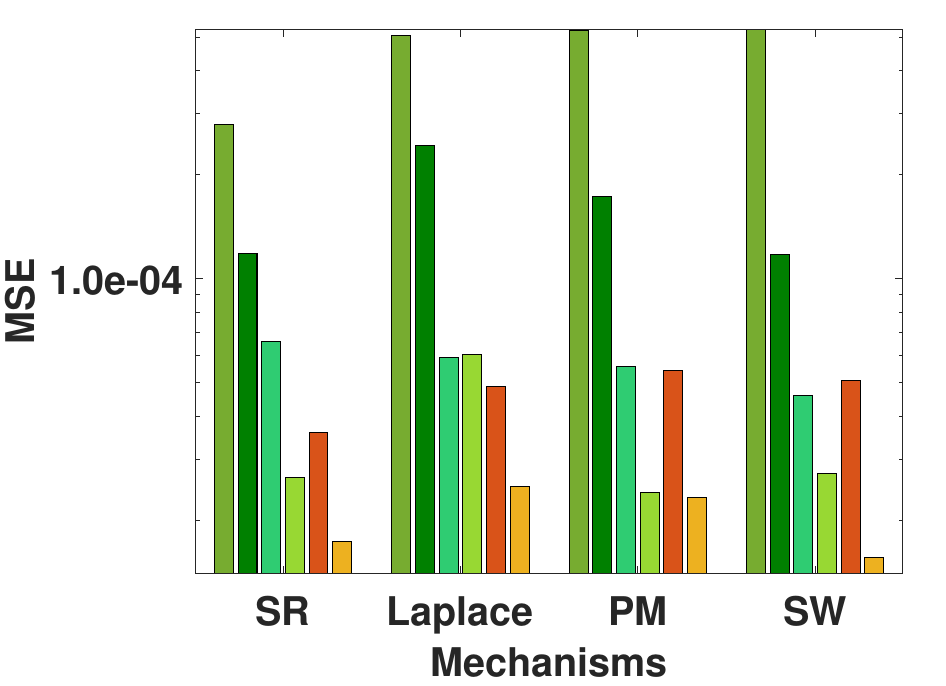}
			\end{minipage}%
		}%
		\subfigure[\textbf{Beta+Sin}, $\epsilon$=0.1, distribution.]{
			\begin{minipage}[t]{0.24\linewidth}
				\centering
				\includegraphics[width=1\textwidth]{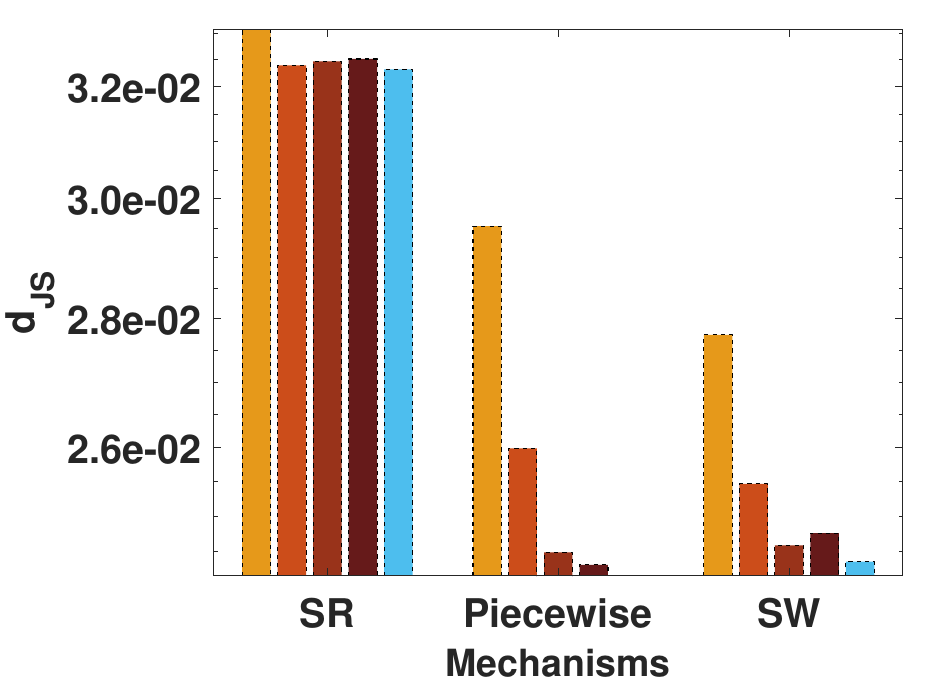}
			\end{minipage}%
		}%
		\subfigure[\textbf{Beta+Sin}, mean.]{
			\begin{minipage}[t]{0.24\linewidth}
				\centering
				\includegraphics[width=1\textwidth]{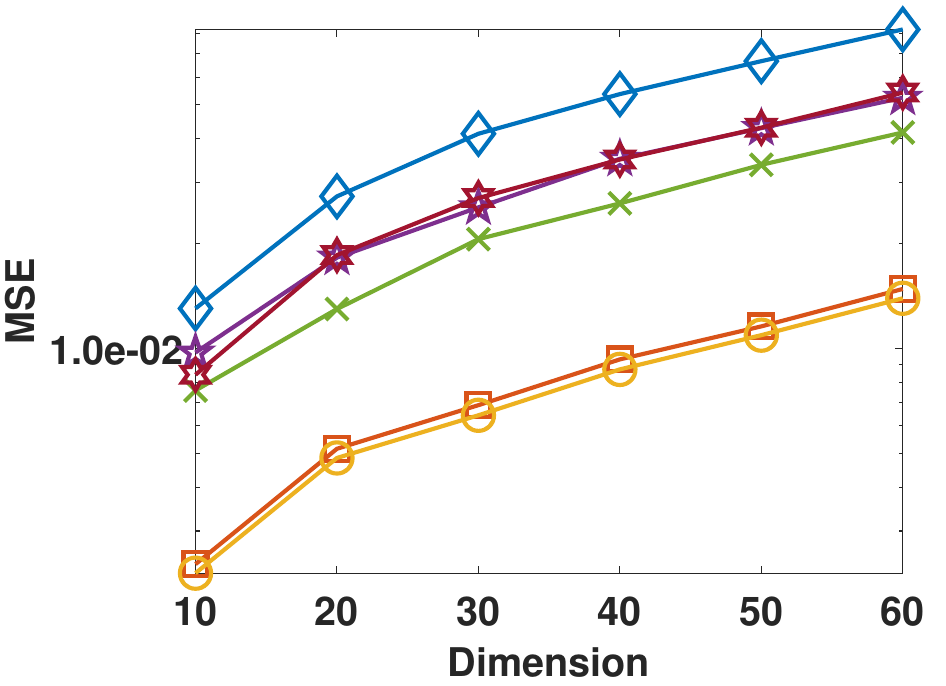}
			\end{minipage}
		}%
		\subfigure[\textbf{Beta+Sin}, distribution.]{
			\begin{minipage}[t]{0.24\linewidth}
				\centering
				\includegraphics[width=1\textwidth]{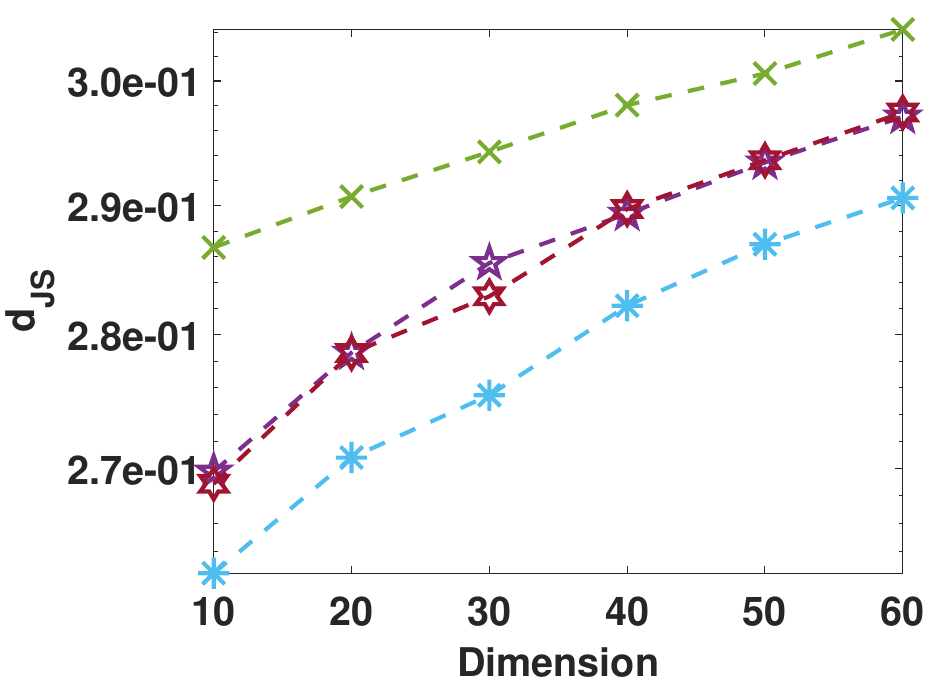}
			\end{minipage}
		}%
		\vspace{-0.05in}
		\caption{Performance evaluation in different scenarios: single mechanism (a)-(b) and high-dimensional cases (c)-(d)}
		\label{Exp3}
		\vspace{-0.15in}
	\end{figure*}

	\noindent\textbf{The Impact of $m$.}
	Figures~\ref{Exp4}(a)-(b) illustrate the impact of service count $m$ on both mean and distribution estimation performance, comparing results between the data collector and individual services. In Figure~\ref{Exp4}(a), which evaluates 4-16 services with privacy budget $\epsilon=0.1$ for each service, both UWA and UA demonstrate significant utility improvements over single-mechanism approaches. This enhanced performance can be attributed to the availability of multiple services, which provide comprehensive information about the posterior  distribution of the original values. Similarly, Figure~\ref{Exp4}(b) presents distribution estimation results across 3-15 services. The results demonstrate that ULE achieves improvements of 0.24\% to 26.04\% over any single service. In Figure~\ref{Exp4}(b), we analyze how the number of services influences ULE's performance improvement. With fewer services ($m$=2,3), ULE shows modest improvements of 0.24\% and 0.62\% respectively, while achieving a 3.62\% improvement over PM at $m$=6. This demonstrates that ULE's performance gains scale with the number of services, as more services provide richer information for distribution estimation.

	\noindent\textbf{The Impact of $n$.}
	Figures \ref{Exp4}(c) (d) illustrate the performance under varying user counts $n$ (from 0.5 to 10 million) with privacy budget $\epsilon=0.1$, where all datasets follow the same distribution as \textbf{Beta+Sin}. As the user count increases, we observe a consistent decrease in both MSE and JS divergence across all methods. This improvement can be attributed to the larger user base size providing more information, thereby reducing variance in the estimates. Our proposed methods - UA, UWA, and ULE - demonstrate superior performance compared to single perturbation mechanisms.
	
	\begin{figure*}[t]
		\hspace{-1in}
		\centering
		{
			\begin{minipage}{6cm}
				\centering
				\includegraphics[scale=0.7]{legend_G1_mean.pdf}
			\end{minipage}
		}
		\hspace{1.2in}
		{
			\begin{minipage}{6cm}
				\centering
				\includegraphics[scale=0.7]{legend_G2_mean.pdf}
			\end{minipage}
		}
		\\
		\hspace{0.4in}
		{
			\begin{minipage}{6cm}
				\centering
				\includegraphics[scale=0.7]{legend_G1_distribution.pdf}
			\end{minipage}
		}
		\hspace{1.1 in}
		{
			\begin{minipage}{6cm}
				\centering
				\includegraphics[scale=0.7]{legend_G2_distribution.pdf}
			\end{minipage}
		}\\
		\centering
		\subfigure[\textbf{Beta25}, mean.]{
			\begin{minipage}[t]{0.24\linewidth}
				\centering
				\includegraphics[width=1\textwidth]{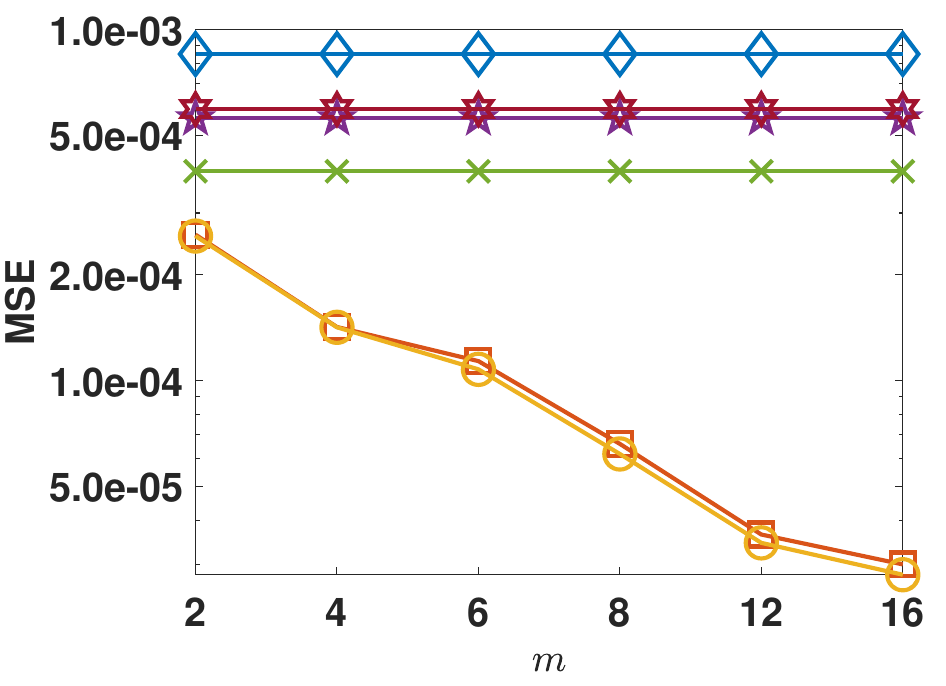}
			\end{minipage}%
		}%
		\subfigure[\textbf{Beta25}, distribution.]{
			\begin{minipage}[t]{0.24\linewidth}
				\centering
				\includegraphics[width=1\textwidth]{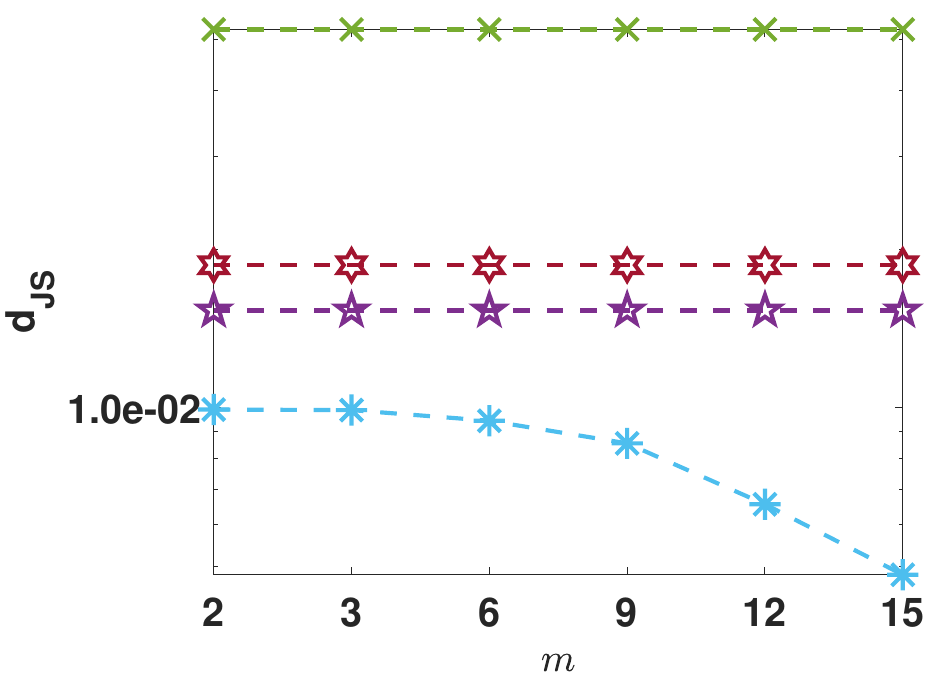}
			\end{minipage}%
		}%
		\subfigure[\textbf{Beta+Sin}, mean.]{
			\begin{minipage}[t]{0.24\linewidth}
				\centering
				\includegraphics[width=1\textwidth]{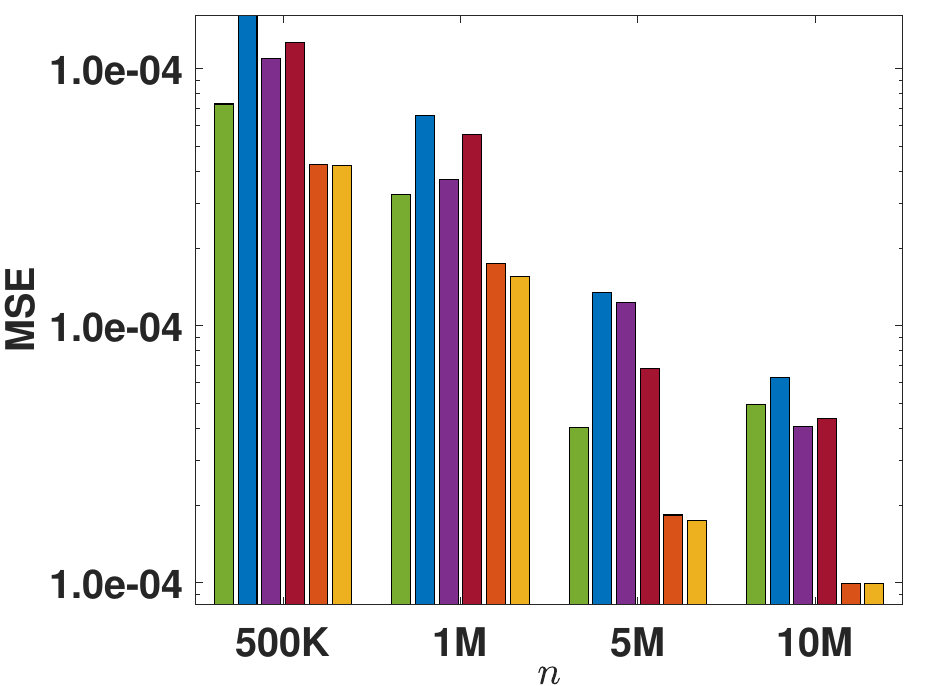}
			\end{minipage}
		}%
		\subfigure[\textbf{Beta+Sin}, distribution.]{
			\begin{minipage}[t]{0.24\linewidth}
				\centering
				\includegraphics[width=1\textwidth]{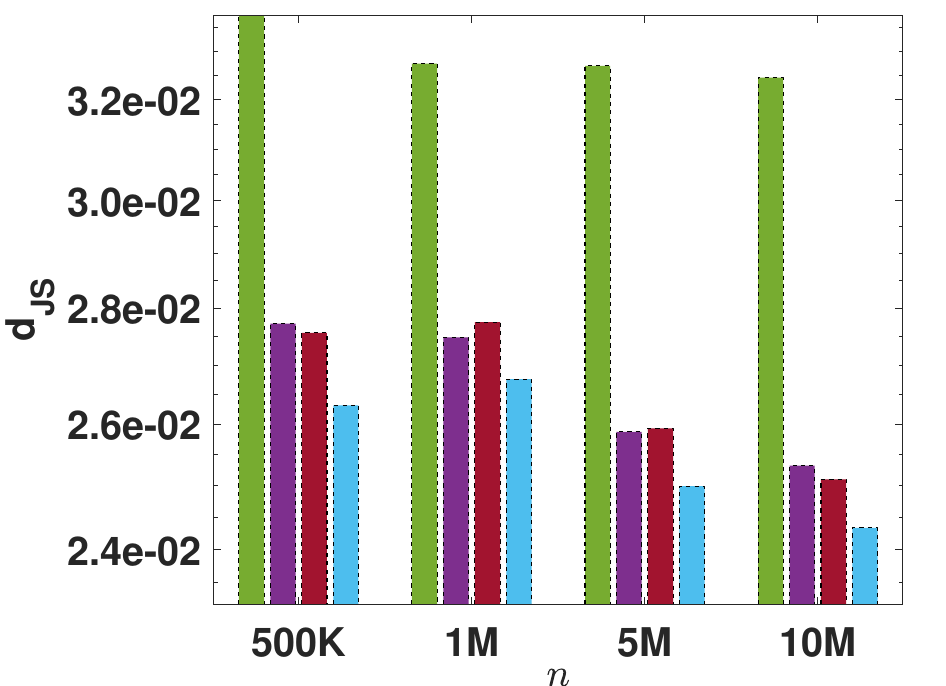}
			\end{minipage}
		}%
		\vspace{-0.05in}
		\caption{Performance evaluation across system scales: server count $m$ (a)-(b) and user base size $n$ (c)-(d)}
		\label{Exp4}
		\vspace{-0.15in}
	\end{figure*}
	
	\subsection{Analysis and Discussion}
	
	\noindent\textbf{Execution Time Analysis.}
	Table~\ref{tab:execution_time} presents the execution time of \textbf{Beta+Sin} with varying numbers of services ($m$) and users ($n$). With $n$ fixed at 1M, UWA exhibits linear growth from 3.45s to 13.16s as $m$ increases from 4 to 16, while ULE maintains relatively stable performance (131.85s to 141.16s) as $m$ increases from 6 to 15. When varying $n$ from 0.5M to 10M (with fixed $m$=4 for UWA and $m$=3 for ULE), both components demonstrate substantial increases in execution time. UWA grows from 2.43s to 349.37s, while ULE shows a more dramatic increase from 69.79s to 12620.94s. When the number of users exceeds 1M, the significant performance degradation of ULE can be attributed to GPU memory constraints, which necessitate matrix partitioning in the EM algorithm, consequently leading to substantially reduced computational efficiency.
	\begin{table}[]
		\caption{Execution Time of \textbf{Beta+Sin} --- UWA for Mean Estimation and ULE for Distribution Estimation (in second)}
		\label{tab:execution_time}
		\begin{tabular}{|c|c|c|c|c|c|c|}
			\hline
			$m$ & \begin{tabular}[c]{@{}c@{}}UWA\\ $n=1M$\end{tabular} & $m$ & \begin{tabular}[c]{@{}c@{}}ULE\\ $n=1M$\end{tabular} & $n$  & \begin{tabular}[c]{@{}c@{}}UWA\\ $m=4$\end{tabular} & \begin{tabular}[c]{@{}c@{}}ULE\\ $m=3$\end{tabular} \\ \hline
			4   & 3.45                                                  & 6   & 131.85                                                        & 0.5M & 2.43                                                 & 69.79                                                        \\ \hline
			8   & 6.73                                                  & 9   & 135.72                                                        & 1M   & 3.45                                                 & 124.05                                                       \\ \hline
			12  & 9.87                                                  & 12  & 138.91                                                        & 5M   & 11.75                                                & 9461.53                                                      \\ \hline
			16  & 13.16                                                 & 15  & 141.16                                                        & 10M  & 349.37                                               & 12620.94                                                     \\ \hline
		\end{tabular}
		\vspace{-0.2in}
	\end{table}
	
	\noindent\textbf{Privacy Budget Analysis.}
	In the LDP setting, privacy budgets exhibit a direct relationship with individual privacy leakage - higher budgets correspond to increasing privacy risks for users. Figures 8(a)-(b) compare the utility and privacy budget consumption of each user under different data collection strategies. In the mean estimation scenario with four services, we evaluate two collection strategies. When the data collector gathers additional information directly from users with a privacy budget of 0.1, the total privacy budget increases to 0.5. When aggregating information from all services without additional queries, the total privacy budget remains at 0.4. The data collector using UWA and UA strategies achieves lower MSE compared to direct collection, regardless of the perturbation mechanism. Similarly, in distribution estimation, the ULE method maintains the lowest JS divergence compared to direct collection, without any additional privacy budget consumption.
	
	\noindent\textbf{A More Dynamic General Case.}
	We further investigate a general scenario where users distribute their information across varying numbers of services, resulting in different user counts across services. In Figures \ref{Exp5}(c)-(d), users are divided into eleven groups (G1-G11): G1-G6 transmit data to two services, G7-G10 to three services, and G11 to all four services. We introduce \textbf{Single best} as the baseline, representing the optimal performance among individual services.
	
	Figure \ref{Exp5}(c) shows that UWA (i.e., the red area) consistently achieves the lowest MSE across all groups, followed by UA (i.e., the orange area), while \textbf{Single best} (i.e., the yellow area) shows higher MSE. Figure \ref{Exp5}(d) compares distribution alignment between \textbf{Single best} (i.e., the green area) and ULE (i.e., the blue area), with ULE showing better performance in most cases. When $m=2$, some groups show minimal improvement due to the limited information provided by each group, since the number of users per group is less than 10,000. These results validate the effectiveness of our methods in more complex and dynamic scenarios with varying user participation patterns.
	
	\noindent\textbf{Effectiveness Analysis.}  For UA, this simple averaging approach performs well when services have equal privacy budgets (Figures 4(a)-(d)), as different mechanisms achieve similar mean estimation performance, making uniform weighting naturally effective. However, when services have different privacy budgets, UA's equal weighting scheme gives excessive influence to high-noise data, leading to performance worse than some individual services (Figures 5(e)-(h) and 6(a)). For UWA, when services have equal privacy budgets, modest improvements are shown over UA as services demonstrate comparable mean estimation performance, resulting in nearly equal weights being assigned by UWA. UWA's slight advantage stems from its Bayesian inference approach. For services with different privacy budgets, UWA significantly outperforms UA by assigning larger weights to low-noise data.
	
	The performance enhancement of ULE over individual services primarily depends on the informativeness of perturbation data from all services. Four key factors influence this informativeness: (1) the number of services - more participating services provide richer information (Figure 7(b)); (2) perturbation mechanism - SW and PM mechanisms preserve more information than SR (Figures 4, 5); (3) privacy budget - larger budgets lead to less perturbation and better information preservation (Figures 5(e)-(h)); and (4) user base size - more users contribute to more comprehensive estimation (Figures 7(d), 8(d)). The more effective information ULE can utilize compared to an individual service, the better the performance improvement becomes.
	
	\begin{figure*}[t]
		\centering
		\subfigure[\textbf{Beta+Sin}.]{
			\begin{minipage}[t]{0.24\linewidth}
				\centering
				\includegraphics[width=1\textwidth]{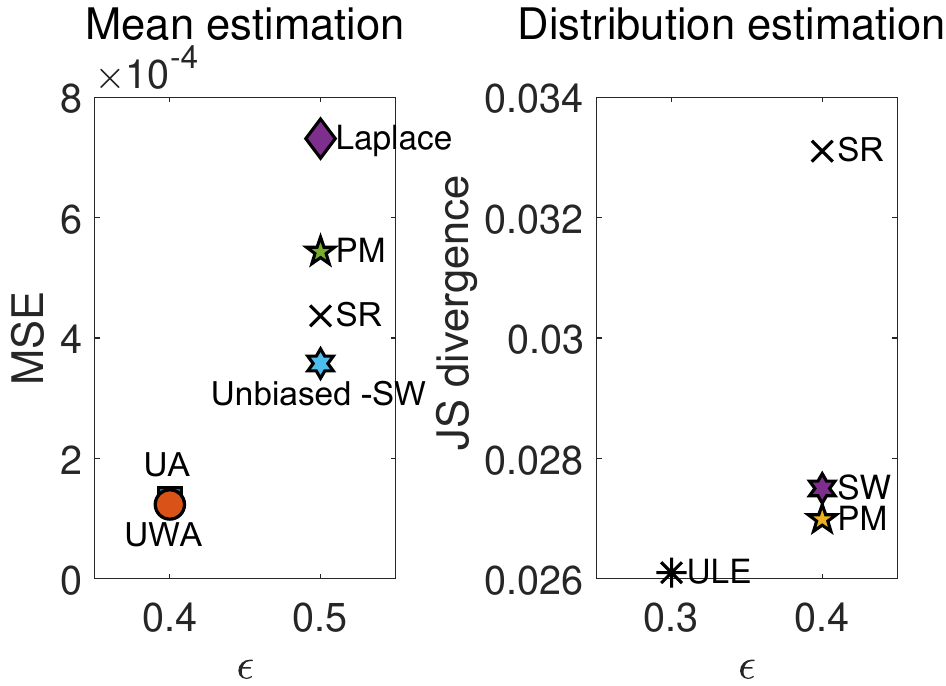}
			\end{minipage}%
		}%
		\subfigure[\textbf{Taxi}.]{
			\begin{minipage}[t]{0.24\linewidth}
				\centering
				\includegraphics[width=1\textwidth]{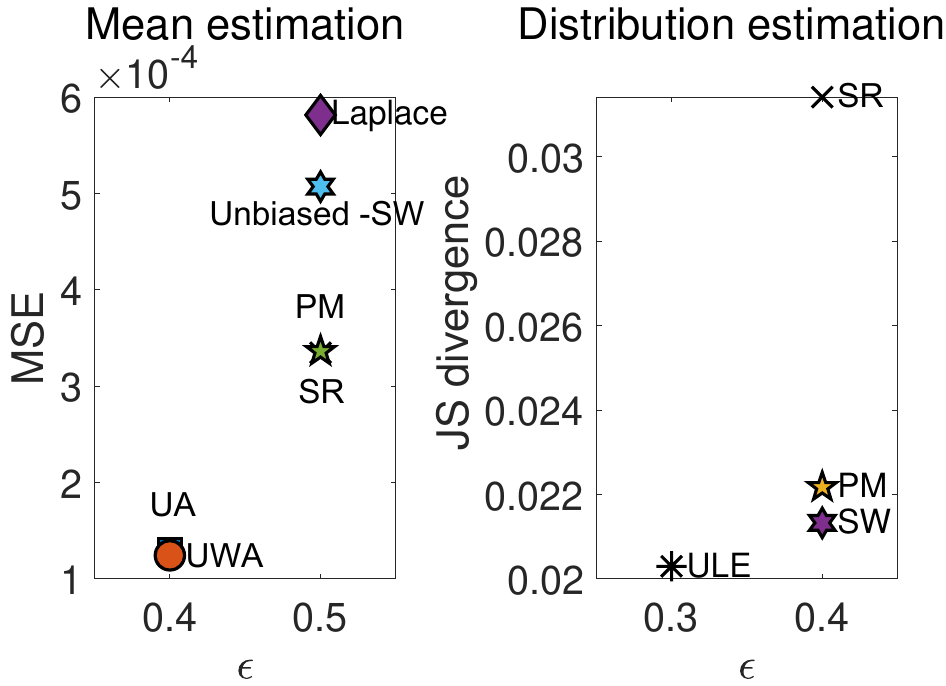}
			\end{minipage}%
		}%
		\subfigure[\textbf{Beta25}, $\epsilon$=0.1, mean.]{
			\begin{minipage}[t]{0.24\linewidth}
				\centering
				\includegraphics[width=1\textwidth]{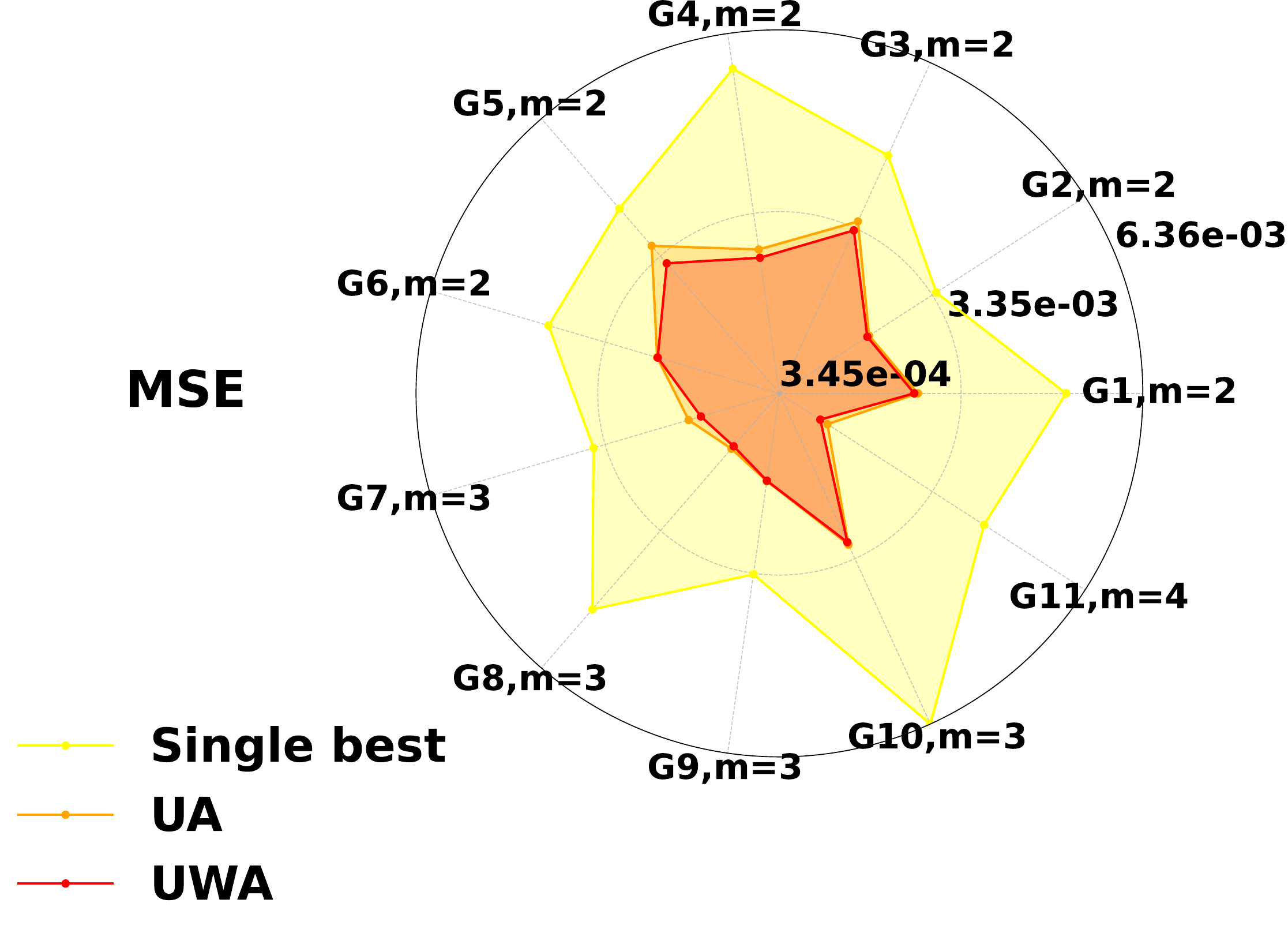}
			\end{minipage}
		}%
		\subfigure[\textbf{Beta25}, $\epsilon$=0.1, distribution.]{
			\begin{minipage}[t]{0.24\linewidth}
				\centering
				\includegraphics[width=1\textwidth]{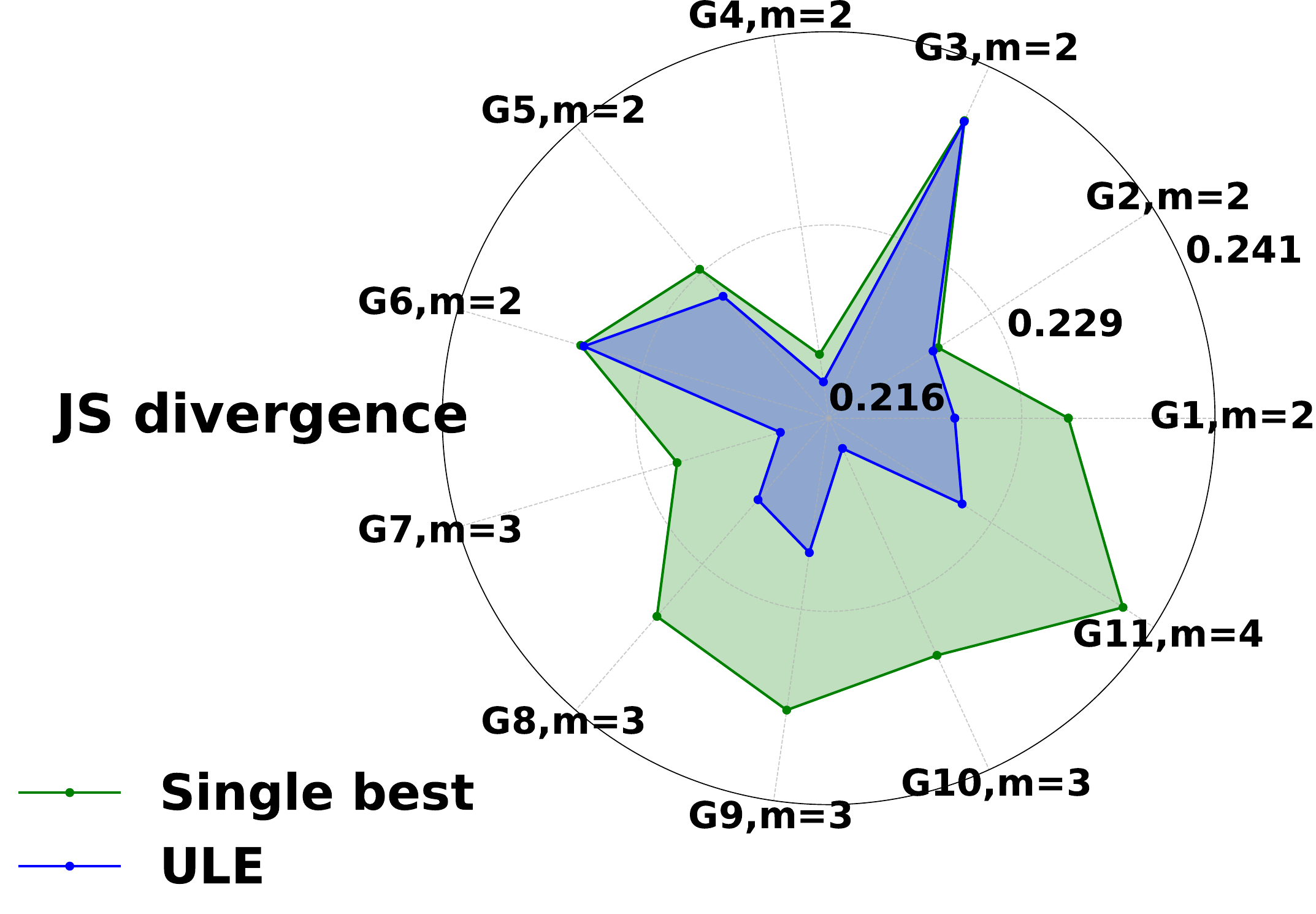}
			\end{minipage}
		}%
		\vspace{-0.05in}
		\caption{ Comparison of utility and privacy guarantees under different data collector strategies (a)-(b), and multi-group performance analysis in generalized implementation scenarios (c)-(d)}
		\label{Exp5}
		\vspace{-0.05in}
	\end{figure*}
	\color{black}

\section{Related Work}
\label{relatedwork}

Local Differential Privacy (LDP)\cite{chen2016private, duchi2013local, kasiviswanathan2011can} extends the concept of Differential Privacy (DP)\cite{dwork2008differential, dwork2006calibrating, mcsherry2007mechanism} to distributed systems, providing robust privacy guarantees for individual users' data. The concept of local privacy originated in 1965 with Warner's~\cite{warner1965randomized} introduction of the randomized response (RR) technique, a simple yet effective perturbation method. Since then, LDP has been extensively studied and widely adopted in various applications. Notable implementations include Google Chrome~\cite{erlingsson2014rappor}, Apple iOS~\cite{team2017learning}, and Microsoft's Windows 10~\cite{ding2017collecting}. Samsung has also conducted research in this area~\cite{nguyen2016collecting}. LDP's applications span diverse domains, including graph data collection~\cite{sun2019analyzing, ye2020towards}, key-value data analysis~\cite{ye2019privkv, gu2020pckv}, and machine learning~\cite{truex2020ldp,wang2023ldp,zhao2020local}.


Recent work by \cite{xu2020collecting} introduces a novel approach to data fusion for multi-source, multi-table data, ensuring differential privacy and protection against frequency attacks. Their research focuses on maintaining strong privacy guarantees during joint data collection, developing methods for joint analysis of perturbed values from various services, and demonstrating OLAP queries on fused data while preserving privacy. However, our work differs significantly in several aspects. While \cite{xu2020collecting} deals with multi-attribute data, we focus specifically on numerical data. Our research does not involve the data collection process itself, instead concentrating solely on the analysis of already perturbed values. Importantly, in our scenario, the same data point may be repeatedly collected by different data collectors, a situation not addressed in \cite{xu2020collecting}. Furthermore, we consider cases where different data collectors may employ distinct LDP perturbation mechanisms, adding another layer of complexity to the aggregation process. These differences highlight our unique contribution in addressing the challenges of aggregating repeatedly collected, differently perturbed numerical data across multiple collectors, thereby extending the field of privacy-preserving data analysis in new directions.

Another category of related work is previous work \cite{yiwen2018utility, du2023differential}. These two papers describe scenarios where different groups of users are allocated inconsistent privacy budgets. Users employ the same perturbation mechanism to distort their data, while the data collector aggregates the perturbed results from different groups. They designed a scheme called Advanced Combination, which combines estimation results from multiple privacy levels. By using least squares method, they find the optimal weights for combination, thereby minimizing utility loss. In paper \cite{yiwen2018utility}, the method only uses the simplest RR perturbation mechanism for binary data, where the variance is unrelated to the original data. In paper \cite{du2023differential}, due to the unknown true user data, they consider the worst-case minimum variance, assigning weights to the mean results of each group. In our work, we do not group users. Additionally, we utilize perturbation knowledge to characterize the possible distribution of true data, assigning weights to the results of different perturbation mechanisms for each user.

\section{Conclusions}
\label{conclusion} 
This paper addresses the growing privacy concerns in the era of multi-service data collection where many services increasingly desire users' information. To mitigate this issue, we propose a novel framework that enables aggregating perturbed data from multiple services while preserving strong user privacy guarantees, including two novel methods for mean estimation, i.e., UA and UWA. Additionally, we develop the ULE method for accurate distribution estimation, which treats the perturbed results from each user as a whole to perform maximum likelihood estimation. The experimental results demonstrate the effectiveness of our framework and the constituent methods.
\section*{Acknowledgements}
\label{sec::ack}
This work was supported by the National Natural Science Foundation of China (Grant No: 92270123 and 62372122), and the Research Grants Council, Hong Kong SAR, China (Grant No:  15208923, 15210023, 15224124, and 25207224).


\bibliographystyle{ACM-Reference-Format}
\bibliography{references}
\end{document}